\newif\ifanonym
\anonymtrue
\anonymfalse

\documentclass[11pt,letterpaper]{article}
\usepackage[margin=1in]{geometry}
\usepackage[T1]{fontenc}

\usepackage{amsmath,amsthm,amssymb}
\usepackage{xspace}
\usepackage{url}
\usepackage{hyperref}
\usepackage[dvipsnames]{xcolor}
\usepackage{algorithm}
\usepackage{algpseudocode}
\usepackage{appendix}
\usepackage{babel}
\usepackage[capitalise,noabbrev]{cleveref}

\usepackage{multirow}
\usepackage{tablefootnote}
\usepackage{afterpage}
\usepackage{lipsum}
\usepackage{booktabs}
\usepackage{thm-restate}
\usepackage{graphicx}
\usepackage{subcaption}
\usepackage{lipsum}

\hypersetup{colorlinks=true,allcolors=blue}

\theoremstyle{plain}
\newtheorem{theorem}{Theorem}[section]
\newtheorem{lemma}[theorem]{Lemma}

\newtheorem{corollary}[theorem]{Corollary}

\newtheorem{fact}[theorem]{Fact}

\theoremstyle{definition}
\newtheorem{definition}[theorem]{Definition}

\theoremstyle{remark}
\newtheorem{remark}[theorem]{Remark}

\newcommand{\ProblemName}[1]{\textsc{#1}}
\newcommand{\oneMedian}{\ProblemName{$1$-Median}\xspace}
\newcommand{\klMedian}{\ProblemName{$(k, \ell)$-Median}\xspace}
\newcommand{\kMedian}{\ProblemName{$k$-Median}\xspace}
\newcommand{\kMeans}{\ProblemName{$k$-Means}\xspace}
\newcommand{\kzC}{\ProblemName{$(k, z)$-Clustering}\xspace}
\newcommand{\onezC}{\ProblemName{$(1, z)$-Clustering}\xspace}

\def\RR{{\mathbb{R}}}

\DeclareMathOperator{\poly}{poly}
\DeclareMathOperator{\err}{err}

\DeclareMathOperator{\cost}{cost}
\DeclareMathOperator{\OPT}{OPT}

\DeclareMathOperator{\dist}{dist}
\DeclareMathOperator{\Diam}{diam}
\DeclareMathOperator{\sdim}{sdim}
\DeclareMathOperator{\vcdim}{vcdim}

\DeclareMathOperator{\ring}{ring}

\newcommand{\calF}{\ensuremath{\mathcal{F}}\xspace}
\newcommand{\calX}{\ensuremath{\mathcal{X}}\xspace}
\newcommand{\calP}{\ensuremath{\mathcal{P}}\xspace}
\newcommand{\calH}{\ensuremath{\mathcal{H}}\xspace}

\newcommand{\Cclose}{\ensuremath{C_{\mathrm{close}}\xspace}}
\newcommand{\Cfar}{\ensuremath{C_{\mathrm{far}}\xspace}}

\newcommand{\pclose}{\ensuremath{p^{\mathrm{close}}\xspace}}
\newcommand{\pfar}{\ensuremath{p^{\mathrm{far}}\xspace}}

\newcommand{\Pclose}{\ensuremath{P_{\mathrm{close}}\xspace}}
\newcommand{\Pfar}{\ensuremath{P_{\mathrm{far}}\xspace}}
\newcommand{\Pmain}{\ensuremath{P_{\mathrm{main}}\xspace}}

\newcommand{\dWSp}{\ensuremath{d^{(p)}_{\mathrm{WS}}}\xspace}

\newcommand{\MyFrame}[1]{\begin{center}\noindent \framebox[\textwidth]{ \begin{minipage}{0.97\textwidth} #1 \end{minipage}}\end{center}}

\title{The Power of Uniform Sampling for Coresets}

\ifanonym
\author{Anonymous Authors}
\else
\author{Vladimir Braverman\thanks{
    Work partially supported by ONR Award N00014-18-1-2364.
    Email: \texttt{vladimir.braverman@gmail.com}
  }\\
  Johns Hopkins University
  \and Vincent Cohen-Addad\thanks{
    Email: \texttt{vcohenad@gmail.com}
  }\\
  Google Research, Switzerland
  \and Shaofeng H.-C. Jiang\thanks{
    Work partially supported by the National key R\&D program of China No. 2021YFA1000900, a startup fund from Peking University,
    and the Advanced Institute of Information Technology, Peking University.
    Email: \texttt{shaofeng.jiang@pku.edu.cn}
  }\\
  Peking University
  \and Robert Krauthgamer\thanks{
    Work partially supported by ONR Award N00014-18-1-2364,
    the Israel Science Foundation grant \#1086/18,
    the Weizmann Data Science Research Center,
    and a Minerva Foundation grant.
    Email: \texttt{robert.krauthgamer@weizmann.ac.il}
  }\\
  Weizmann Institute of Science
  \and Chris Schwiegelshohn\thanks{
  Work partially supported by an Independent Research Fund Denmark (DFF) Sapere Aude Research Leader grant No 1051-00106B.
    Email: \texttt{cschwiegelshohn@gmail.com}
  }\\
  Aarhus University
  \and Mads Bech Toftrup\thanks{
    Email: \texttt{toftrup@cs.au.dk}
  }\\
  Aarhus University\and Xuan Wu\thanks{
    Email: \texttt{wu3412790@gmail.com}
  }\\
  Johns Hopkins University
}
\fi

\begin{document}

    \begin{titlepage}
        \maketitle
        \thispagestyle{empty}
        \begin{abstract} 
Motivated by practical generalizations of the classic $k$-median and $k$-means objectives,
such as clustering with size constraints, fair clustering, and Wasserstein barycenter,
we introduce a meta-theorem for designing coresets for constrained-clustering problems.
The meta-theorem reduces the task of coreset construction to one on a bounded number of ring instances with a much-relaxed additive error.
This reduction enables us to construct coresets using uniform sampling,
in contrast to the widely-used importance sampling,
and consequently we can easily handle constrained objectives.
Notably and perhaps surprisingly, this simpler sampling scheme can yield
coresets whose size is independent of $n$, the number of input points.

Our technique yields smaller coresets, and sometimes the first coresets,
for a large number of constrained clustering problems,
including capacitated clustering, fair clustering,
Euclidean Wasserstein barycenter, clustering in minor-excluded graph,
and polygon clustering under Fr\'{e}chet and Hausdorff distance.
Finally, our technique yields also smaller coresets
for $1$-median in low-dimensional Euclidean spaces,
specifically of size $\tilde{O}(\varepsilon^{-1.5})$ in $\mathbb{R}^2$
and $\tilde{O}(\varepsilon^{-1.6})$ in $\mathbb{R}^3$.
\end{abstract}

     \end{titlepage}

\section{Introduction}

Over the last 20 years, coresets and in particular coresets for clustering problems have received substantial attention.
At a high a level,
a coreset for a data set $P$ with respect to a set of queries $\mathcal{Q}$
with query-evaluation function $f_P: \mathcal{Q} \rightarrow \mathbb{R}_{+}$,
is a data set $S$ with a corresponding $f_S: \mathcal{Q} \rightarrow \mathbb{R}_{+}$
that approximates the evaluation function for every query.
Typically, $S$ is a (small) reweighted subset of $P$,
and the function $f_S$ is defined similarly to $f_P$.
For many clustering problems, $P$ is the input set,
each query $q\in \mathcal{Q}$ is a candidate center set
and its corresponding $f_P(q)$ is the cost induced by this center set,
hence a coreset is just a smaller instance ($S$ instead of $P$)
of the same clustering problem.

The quality of a coreset $S$ is usually measured by its size
(the number of distinct points)
and by the type of query evaluations that it approximates.
For example, a natural query for a clustering problem
is the cost induced by greedily assigning every point in $P$
to its closest neighbor in the center set $q$,
aggregated over the points in $P$.
A prime example is the Euclidean $(k,z)$-clustering problem in dimension $d$,
in which $P\subset \mathbb{R}^d$ is the input,
$z>0$ is a parameter (typically fixed),
each query is a center set $C \subset \mathbb{R}^d$ of size $k$,
and the query evaluation is the cost function
\begin{equation}
  \label{eq:kzclustering}
  \cost_z(P, C) := \sum_{x\in P} {(\dist(x, C))^z},
\end{equation}
where $\dist(x, y) := \|x - y\|_2$ and $\dist(x, C) := \min_{c \in C} \dist(x, c)$.
The special case $z=2$ is the widely studied Euclidean \kMeans problem.
Following a long line of research~\cite{DBLP:conf/stoc/Har-PeledM04,DBLP:conf/stoc/FrahlingS05,DBLP:journals/dcg/Har-PeledK07,Chen09,langberg2010universal,DBLP:conf/stoc/FeldmanL11,FichtenbergerGSSS13,FSS20,BLLM16,DBLP:conf/icml/BravermanFLSY17,sohler2018strong,BachemLL18,BecchettiBC0S19,huang2020coresets,BJKW21,Cohen2022Towards},
it is now known that Euclidean \kMeans admits an $\varepsilon$-coreset $S$ of size $\tilde{O}(k \varepsilon^{-2} \min\{k,\varepsilon^{-2}\})$ \cite{cohen2021new,Cohen2022Towards},
where an \emph{$\varepsilon$-coreset} means that for every center set $C$,
the cost of $P$ and that of $S$ are within a $(1\pm \varepsilon)$-factor.

The most immediate approach to construct a coreset is to
sample a subset of the input (and reweight its points appropriately),
and the main challenge is to find a sampling distribution that works well.
A natural starting point is \emph{uniform sampling},
however without further assumptions or preprocessing steps,
it is easy to construct instances where uniform sampling
requires so many samples that it achieves no significant space savings.\footnote{Consider a one-dimensional input $P\subset\mathbb{R}$
  with $n$ points at $0$ and a single point at $1$.
  For $k=1$, by placing a center at $0$, the only point incurring a positive cost is the point at $1$.
  However, a uniform sample is unlikely to pick the point at $1$ unless the sample size is $\Omega(n)$.
}

Instead, state-of-the-art algorithms for computing coresets
are typically based on non-uniform sampling.
These algorithms, initiated by Chen's seminal paper \cite{Chen09} and codified in their modern form under the name \emph{sensitivity sampling} by Feldman and Langberg~\cite{DBLP:conf/stoc/FeldmanL11},
draw $|S|$ points from the same probability distribution $(p_x: x\in P)$,
and reweight every sample $x$ inverse proportionally to its sampling probability,
namely, $w_S(x) = \frac{1}{|S|\cdot p_x}$.
The sampling probability $p_x$ is set proportionally to the \emph{sensitivity} of $x$,
which is the maximum possible relative contribution of $p$ to any query evaluation.
For example, for Euclidean \kMeans, this is
$s(x):=\operatorname{sup}_{|C|=k} \frac{(\dist(x, C))^2}{\cost(P,C)}$.
The sensitivity sampling framework has become an enormously successful and popular method for many additional problems,
including kernel methods \cite{PhillipsT20,JKLZ21}, low-rank approximation \cite{maalouf2019fast}, linear regression \cite{HuangSV20,TukanMF20}, and logistic regression \cite{MunteanuSSW18,MRM21}.

Unfortunately, not all problems are easily expressed in the sensitivity framework.
Consider, for example, clustering with size constraints,
which loosely means constraining the number of points served by every center.
For example, limiting the centers to each serve at most $T$ points 
is known as capacitated clustering with uniform capacity $T$.
Constrained clustering introduces a number of technical issues that make it difficult to generalize
the analysis for coresets in the unconstrained setting.
Perhaps the most glaring obstacle is that sensitivity sampling distorts
the total weight of the points (it is preserved only in expectation).
While it is easy to preserve the total weight
by rescaling the weights in $S$ so that $\sum_{x \in S} w_S(x) = |P|$,
size constraints usually require the total weight to be preserved for many subsets,
which cannot be achieved under the same scaling.
Indeed, directly applying the sensitivity sampling framework to capacitated clustering can result in additive error proportional to the diameter, which is generally unaffordable.

In contrast, uniform sampling can avoid the aforementioned issue,
by running it on top of some preprocessing, like Chen's~\cite{Chen09} metric decomposition, 
and indeed it has been applied to obtain coresets for size-constrained clustering problems,
including capacitated and fair clustering~\cite{cohen2019fixed,DBLP:conf/icalp/BandyapadhyayFS21}.
While uniform sampling only yields a coreset with additive error
for each part in the decomposition,
this additive error can essentially compensate for the weight-distortion issue.
However, the framework of \cite{Chen09} also has a number of drawbacks
compared to the subsequent sensitivity-sampling methods.
For example, its coreset size grows (at least linearly) with $\log |P|$,
regardless of the metric space and objective function.
In contrast, sensitivity sampling, when applicable,
can yield coreset size that depends only on $k$ and $\varepsilon$.
Thus, improving our ability to apply uniform sampling in coreset constructions has been an important open problem for more than a decade.

\subsection{Our Results}

We propose an improved coreset framework that preprocesses the data
so that uniform sampling is applicable.
Similarly to Chen's method~\cite{Chen09},
the key step is a reduction to ring instances, defined as follows.
A point set $R$ is called a \emph{ring}
if all its points are at distance in the range $[r, 2r]$
from some center point $c$ (for some $r>0$).
Our main result is the following meta-theorem (see Theorem~\ref{thm:reduct_ring} for a formal statement):\MyFrame{
  Assume that for rings, uniform sampling produces a coreset of size $T$ with \emph{additive} error at most $\varepsilon r |R|$;
  then for every input $P$ one can construct a coreset of size $O(T\cdot k^2/\varepsilon)$.
}
This result generalizes to \kzC,
as defined in~\eqref{eq:kzclustering}
and more formally in \Cref{def:kzc}.

This result has a number of applications.
To begin with, it allows us to obtain the first coresets whose size is \emph{independent of $|P|$} for the aforementioned problems of clustering with size constraints.
This includes:\begin{itemize}
\item A coreset of size $\poly(k/\varepsilon)$
  for a \emph{capacitated} version of Euclidean \kMedian and \kMeans (see \Cref{sec:prelim} for the definition),
  which improves over the size bound $\poly(k/\varepsilon\cdot \log |P|)$
  of Cohen-Addad and Li~\cite{cohen2019fixed}.
  See Theorem~\ref{thm:capC} for details.
\item A coreset of size $\poly(k/\varepsilon)$
  for \emph{fair} Euclidean \kMedian and \kMeans (see \Cref{sec:prelim} for the definition),
  which improves over the size bounds $\poly(k/\varepsilon\cdot \log |P|)$
  of Bandyapadhyay, Fomin and Simonov~\cite{DBLP:conf/icalp/BandyapadhyayFS21},
  and $\poly(k/\varepsilon^d)$
  of Huang, Jiang and Vishnoi~\cite{HJV19}.
See Theorem~\ref{thm:fairC} for details.
\end{itemize}

Moreover, our framework has applications to other clustering problems.
On a technical level,
a natural approach to proving that a randomly chosen subset $S$ is a coreset
is to first make sure that, with high probability, $S$ preserves the cost for a \emph{single} arbitrary center set,
and then apply a union bound over all possible center sets.
Since the number of possible center sets can be huge, and even infinite,
the space is often discretized by a certain ``net'' before applying a union bound.
For instance,
a recent approach established an $\varepsilon$-approximate centroid set,
as defined by~\cite{cohen2021new};
however, only a few techniques are known to find
such an $\varepsilon$-approximate centroid of small size.
A popular alternative to this net approach is to bound VC-dimension\footnote{Strictly speaking, the correct term here is pseudo-dimension,
  because VC-dimension is defined for a range space (i.e., set system),
  and we use here VC-dimension as a generic term for sake of exposition.
  The technical sections use the correct mathematical terminology,
  which is often the shattering dimension. It is well-known that these three terms are closely related.
}
of the function space $\{f_x(C):=w(x) \cdot \dist(x,C)\}_x$,
where $w(x)$ is related to the probability of sampling $x$.
In general, this $w(x)$ is non-uniform over all $x\in P$,
and this is particularly true for coresets constructed via sensitivity sampling.
Let us refer to the special case where $w(x)$ is uniform
(i.e., takes a single value over all $x\in P$)
as the \emph{uniform} function space.
Our framework for coreset construction is based on uniform function spaces,
which turns out to be a crucial difference with major advantages
compared to the sensitivity-sampling framework.
Indeed, the advantage of relying on uniform function spaces is two-fold.
\begin{enumerate}
\item
  For several problems,
  we know how to bound the VC-dimension of the uniform function space,
  but not that of the non-uniform function space.
  Examples include the shortest-path metric in planar graphs \cite{DBLP:journals/dm/BousquetT15}
  and the Fr{\'e}chet distance~\cite{DBLP:journals/dcg/DriemelNPP21}.
Our new framework leads to new/improved coreset results for such clustering problems.
\item
  The uniform function space has a simpler structure
  and may have a smaller VC-dimension bound.
  Consider for instance the widely studied range space induced by halfspaces in the Euclidean plane ($\RR^d$ for $d=2$);
  the VC-dimension of its uniform range space is known to be exactly $d+1=3$,
  whereas for the non-uniform range space the known upper bound is only
  $3d+1=7$~\cite[Lemma 3.3]{NEURIPS2021_90fd4f88}.
  This leads directly to better bounds on the coreset size.
In particular, when the VC-dimension is low,
  one can plug in at a key step of the analysis,
  a bound from discrepancy theory~\cite{chazelle2001discrepancy}
  about $\varepsilon$-approximation,
  which beats the usual $\varepsilon^{-2}$ factor.
\end{enumerate}

These advantages lead to new coreset results in several different metric spaces:
\begin{itemize}
\item
  A coreset of size $\tilde{O}(\varepsilon^{-1.5})$ for geometric median in dimension $2$,
  and of size $\tilde{O}(\varepsilon^{-1.6})$ for dimension $3$ (see \Cref{cor:lowdim}).
  The previously known coreset size for these problems was $\tilde{O}(\varepsilon^{-2})$ due to \cite{DBLP:conf/stoc/FeldmanL11}.

\item
  A coreset of size $\tilde{O}(\varepsilon^{-2} d \ell)$
  for the $p$-Wasserstein barycenter (see \Cref{thm:barycenter}).
  This is the \kMedian problem for $k=1$,
  in a metric space over all probability distributions that are supported on at most $\ell$ points in $\RR^d$.
  The $p$-Wasserstein distance between two distributions $D_1$ and $D_2$
  is the $p$-th moment of the minimum cost matching between the distributions
  (i.e., edge weights represent Euclidean distance raised to power $p$,
  and the total cost is raised to power $1/p$).
This improves over the previous bound $O(\varepsilon^{-2}d^4 \ell^8)$,
  due to~\cite{izzo2021dimensionality}.

\item
  A coreset of size $\tilde{O}(|H|\cdot \text{poly}(k/\varepsilon))$
  for \kMedian in shortest-path graph metrics that are induced by
  graphs excluding a fixed minor $H$ (see \Cref{sec:more}).
  This improves over a previous bound $\tilde{O}(f(|H|)\cdot \poly(k/\varepsilon))$,
  due to~\cite{BJKW21},
  where $f$ is not specified but is at least doubly exponential.

\item
  A coreset of size $\poly(kd\ell/\varepsilon \cdot \log m)$
  for \kMedian under Fr\'{e}chet and Hausdorff distances (see \Cref{sec:more}).
  In this problem, also known as \klMedian,
  the data set comprises of polygonal curves in $\mathbb{R}^d$, each with at most $m$ line segments,
  and the center curves are restricted to at most $\ell$ line segments.
  This is the first coreset whose size is independent of the number of input curves, improving over \cite{BR22}.
\end{itemize}
These new results highlight the flexibility of our framework
and we expect that it will have additional applications.

\subsection{Our Techniques}

We outline our main technical novelty in obtaining the meta-theorem
(formalized in Theorem~\ref{thm:reduct_ring}
that reduces the coreset-construction problem into only $\tilde{O}(k^2/\epsilon)$ ring instances,
in which uniform sampling is applicable.

For sake of presentation, let us focus on $z = 1$ (i.e., \kMedian).
The proof of the meta-theorem combines
several known geometric techniques for constructing coresets,
that originally cannot give a coreset with size bound $\poly(k/\epsilon)$.
Our algorithm first finds an $(O(1),O(1))$-bicriteria approximation $C^*$ with $|C^*|=O(k)$ centers\footnote{An $(\alpha,\beta)$-bicriteria approximation for a clustering problem
  is a set of at most $\beta\cdot k$ centers that has cost $\alpha\cdot \OPT_k$,
  where $\OPT_k$ is the optimal cost of clustering using $k$ centers.
},
then partitions the data accordingly into $O(k)$ clusters,
and then further partitions each cluster into rings with exponentially-increasing radii,
similarly to the steps in~\cite{Chen09}.
The issue with this partition, as noted also in \cite{Chen09},
is that it creates $O(\log n)$ rings,
which eventually introduces an $O(\log n)$ factor in the coreset size.
To bypass this, we identify in each cluster
a set of $\tilde{O}(k / \epsilon)$ high-cost rings
(and thus $\tilde{O}(k^2/\epsilon)$ rings in total),
for which the points inside contribute significantly to the objective.
Call these high cost rings \emph{marked}, and the remaining rings \emph{unmarked}.
Consecutive unmarked rings (i.e., between two marked rings)
are merged into in at most $\tilde{O}(k / \epsilon)$ \emph{unmarked groups}.
The $\tilde{O}(k^2/\epsilon)$ marked rings are handled
as in~\cite{Chen09} using uniform sampling.
The remaining issue is how to construct coresets for the unmarked groups.
An unmarked group can be a union of multiple consecutive rings,
and since points do not have a similar distance to the cluster center,
uniform sampling is no longer applicable.
However, by our construction, each unmarked group has a small contribution to the cost and we show that a simple two-point geometric construction can already serve as a coreset for the entire group.
Such a two-point coreset is much more powerful than it appears to be.
In particular, it even satisfies a property that we call assignment-preserving
(see \Cref{def:assign:add} and a similar formulation in prior work~\cite{DBLP:conf/waoa/0001SS19,HJV19,DBLP:conf/icalp/BandyapadhyayFS21}),
and hence can serve as a coreset for clustering with capacity and fairness constraints.

Let $c_i\in C^*$ and let $P_i$ denote the cluster with center $c_i$.
Technically, the construction of the unmarked groups and their two-point coresets
is done by interpreting the entire cluster as a one-dimensional instance
(by taking $\dist(x, c_i)$ for each point $x \in P_i$),
and then applying on the unmarked rings
a known greedy-bucketing construction for dimension one \cite{DBLP:journals/dcg/Har-PeledK07}.
To construct the two-points coreset for a group $G$,
let $\pclose,\pfar\in G$ be a closest point and a furthest point, respectively,
from the center $c_i$.
Then for every point $x \in G$, represent the distance $\dist(x, c_i)$
as a convex combination of $\dist(\pclose, c_i)$ and $\dist(\pfar, c_i)$,
namely, find $\lambda_x\in [0,1]$ such that
$\dist(x,c_i)=\lambda_x\dist(\pclose,c_i) + (1-\lambda_x)\dist(\pfar,c_i)$.
Now let the coreset for $G$ be
$S:=\{\pclose,\pfar\}_w$
with weights $w(\pclose)=\sum_{x\in G} \lambda_x$ and $w(\pfar)=\sum_{x\in G} (1-\lambda_x)$.
Obviously, $S$ has only two distinct points
and it preserves the total weight and the cost with respect to $c_i$
as the entire $G$.

It remains to analyze the error between our two-point coreset $S$ and the group $G$ with respect to an arbitrary center set, even with capacity constraints.
Fix a center set $C$ with $|C|=k$ and capacity constraint $\Gamma : C \to \mathbb{R}_+$
that prescribes the number of points connected to each center $c\in C$
(see Definition~\ref{def:assign} for formal definition).
We first observe that due to the triangle inequality and our grouping method,
the cost of clustering $S$ approximates that of $G$ within an additive error,
namely,
$|\cost(G, C, \Gamma)-\cost(S,C,\Gamma)| \leq \tilde{O}(\frac{\epsilon}{k})\cdot \cost(P_i,c_i)$
(see Definition~\ref{def:constraint} and Lemma~\ref{lemma:fair:color}).
However, as the cluster $P_i$ has $\tilde{O}(\frac{k}{\epsilon})$ unmarked groups,
its cumulative error is bounded by $\tilde{O}(\frac{k}{\epsilon})\cdot \tilde{O}(\frac{\epsilon}{k})\cdot \cost(P_i,c_i) = \tilde{O}(\cost(P_i,c_i))$,
which exceeds our intended error bound $\tilde{O}(\epsilon) \cdot \cost(P_i, c_i)$.
To reduce the number of groups that can suffer an additive error,
we further divide the unmarked groups into colored
groups and uncolored groups with respect to $C$.
In particular,
we call a ring ``important'' if it contains any center from $C$.
We ``color'' $O(\log \tfrac{1}{\epsilon})$ neighboring rings of each important ring
and ``color'' all the groups that contain at least one colored ring.
This way, we obtain at most $O(k\log \tfrac{1}{\epsilon})$ colored groups.
We let these $O(k\log \tfrac{1}{\epsilon})$ colored groups suffer the additive
error, and this time the total error from them is bounded by $\tilde{O}(\epsilon)\cdot \cost(P_i, c_i)$.

It remains to bound the error for the uncolored groups,
and crucially, in Lemma~\ref{lemma:fair:uncol}
we show these groups do not suffer an additive error but only a multiplicative error.
A key observation is that if a group $G$ is not colored (with respect to $C$),
then every $c \in C$ is either too far from all the points in $G$
or too close to the cluster center $c_i$.
Based on this observation, we surprisingly find that when the group is not colored, our simple two-points coresets $S$
can already serve as an assignment-preserving coreset \emph{without} additive error.
This Lemma~\ref{lemma:fair:uncol} is one of the main technical lemmas
that deal with the assignment constraint,
and its proof requires very careful explicit constructions
for the assignments of the two-point coreset $S$ and the group $G$.

\subsection{Additional Related Work}

Although the coreset paradigm is most often applied to clustering problems,
there are actually several other applications,
see the surveys \cite{DBLP:journals/ki/MunteanuS18,DBLP:journals/widm/Feldman20} for further pointers to the literature.
Restricting attention to coresets for clustering,
the most common setting is that of a Euclidean space,
but there are many results also for other metric spaces.
To streamline the presentation, we focus here on the results for $\kMedian$.
For general $n$-point metrics, \cite{DBLP:conf/stoc/FeldmanL11} gave coresets
of size $O(\frac{k\log n}{\epsilon^2})$,
and for general metrics with bounded doubling dimension $d$,
\cite{huang2018epsilon} designed a coreset of size $O(\frac{k^3d}{\epsilon^2})$,
which was later improved by \cite{cohen2021new} to $\tilde{O}(\frac{kd}{\epsilon^2})$.
Another line of research addresses the shortest-path metrics of graphs,
and notably, $\poly(k/\epsilon)$-size coresets for \kMedian were obtained
for graphs of bounded treewidth, planar graphs, and more generally excluded-minor graphs
\cite{baker2020coresets,BJKW21,cohen2021new}.
For an empirical evaluation of these algorithms, we refer to \cite{SchwiegelshohnS22}.

Coresets for even more general clustering problems, i.e., beyond \kzC,
received significant attention as well.
Apart from the capacity and fairness constrained clustering that are studied in this paper,
coresets were designed also for ordered weighted clustering \cite{braverman2019coresets},
for clustering with outliers \cite{huang2018epsilon,Ding2019GreedySW},
for training Gaussian mixture models \cite{DBLP:conf/kdd/BachemL018,feldman2019coresets},
for time-series clustering \cite{NEURIPS2021_c115ba9e}, and many other related problems.
Another interesting generalization is clustering of sets of points in $\RR^d$
(instead of points),
including arbitrary finite sets \cite{pmlr-v119-jubran20a},
lines \cite{DBLP:conf/nips/MaromF19},
and axis-align affine subspaces \cite{NEURIPS2021_90fd4f88}.

     \section{Preliminaries}
\label{sec:prelim}

\paragraph{Notations.}
We use $\mathbb{R}_+$ to denote set $\{x\geq 0\mid x\in \mathbb{R}\}$. A weighted set $S$ is associated with a weight function $w_S : S \to \mathbb{R}_+$.
We interpret an unweighted set $S$ as a weighted set with unit weight, i.e., $w_S(\cdot) = 1$.
For some weight function $w_S : S \to \mathbb{R}_+$ and $T \subseteq S$,
define $w_S(T) := \sum_{x \in T}{w_S(x)}$.
We assume there is an underlying metric $M(X, \dist)$ throughout the paper.
This metric space may not be finite; for instance, it can be Euclidean space $(\mathbb{R}^d, \ell_2)$.
For a point $x \in X$ and a point set $C \subseteq X$, let $\dist(x, C) := \min_{c \in C}{\dist(x, c)}$.
For $u \in X, 0 \leq a < b$,
let $\ring(u, a, b) := \{ x \in X : a < \dist(x, u) \leq b\}$
be the set of points within distance between $a$ and $b$ from $u$.

We need the following generalized triangle inequalities which are well-known tools for studying \kzC. Variants of these inequalities can be found in multiple related papers \cite{MakarychevMR19,cohen2021new,FSS20,sohler2018strong}.

\begin{lemma}[Generalized triangle inequality]
\label{lem:gen:tri}
Let $a,b,c \in X$ and $z\geq 1$. For every $0<t\leq 1$, the following inequalities hold.
\begin{enumerate}
\item (Corollary A.2 of \cite{MakarychevMR19})
$$
    \dist(a,b)^z \leq (1+t)^{z-1} \dist(a,c)^z + \big(1+\frac{1}{t}\big)^{z-1} \dist(b,c)^z
$$
\item (Claim 5 of \cite{sohler2018strong})
$$
    |\dist(a,c)^z-\dist(b,c)^z|\leq t\cdot \dist(a,c)^z + (\frac{3z}{t})^{z-1}\dist(a,b)^z.
$$
\end{enumerate}
\end{lemma}

\begin{definition}[Coresets for \kzC]
    \label{def:kzc}
    Given a weighted data set $P \subseteq X$,
    for $C \subseteq X$ with $|C| \leq k$,
    define the cost for \kzC as
    $$
        \cost_z(P, C) := \sum_{x \in P} {w_P(x) \cdot (\dist(x, C))^z}.
    $$
    For $0 < \epsilon < 1$,
    a weighted set $S$ such that $S \subseteq P$
    is an $\epsilon$-coreset for \kzC if
    \begin{align*}
        \forall C \subseteq X, |C| \leq k,\qquad
        \cost_z(S, C) \in (1 \pm \epsilon) \cdot \cost_z(P, C),
    \end{align*}
\end{definition}

The following definition of assignment constraints generally captures the constraints in fair clustering and capacitated clustering,
and our key notion of assignment-preserving coresets is defined with respect to it.
Similar notions of assignment constraints and assignment-preserving coresets have also been considered in previous works which study fair clustering~\cite{DBLP:conf/waoa/0001SS19,HJV19,DBLP:conf/icalp/BandyapadhyayFS21}.
\begin{definition}[Assignment constraints and assignment functions]
    \label{def:assign}
    Given a weighted set $P \subseteq X$ and $C \subseteq X$,
    an assignment constraint is a function $\Gamma : C \to \mathbb{R}_+$
    such that $\sum_{c \in C}{\Gamma(c)} = w_P(P)$,
    and we call an assignment function $\sigma : P \times C \to \mathbb{R}_+$
    consistent with $\Gamma$, denoted as $\sigma \sim \Gamma$,
    if $\forall c \in C$,
    $\sigma(P, c) := \sum_{p \in P}{\sigma(p, c)} = \Gamma(c)$. For $P_1\subseteq P$ and $C_1\subseteq C$, we define
       \begin{align*}
        \cost_z^\sigma(P_1, C_1) :=
        \sum_{x \in P_1}\sum_{c \in C_1} \sigma(x, c) \cdot (\dist(x, c))^z
    \end{align*} as the connection cost between $P_1$ and $C_1$ under $\sigma$.

\end{definition}

\begin{definition}[\kzC with assignment constraints]
    \label{def:constraint}
    Given a weighted data set $P \subseteq X$,
    a center set $C \subseteq X$ with $|C| \leq k$,
    and an assignment constraint $\Gamma : C \to \mathbb{R}_+$
    the objective for \kzC with assignment constraint $\mu$ is defined as
    \begin{align*}
        \cost_z(P, C, \Gamma) := \min_{ \sigma : P \times C \to \mathbb{R}_+,
        \sigma \sim \Gamma } \cost_z^\sigma(P, C).
    \end{align*}
\end{definition}

\begin{definition}[Assignment-preserving coresets for \kzC]
    \label{def:assign_coreset}
    Let $P$ be a weighted dataset. A weighted subset $S \subseteq P$
    is an assignment-preserving $\epsilon$-coreset for \kzC,
    if $w_P(P) = w_S(S)$, and for every $C \subseteq X$ with $|C| \leq k$
    and assignment constraint $\Gamma : C \to \mathbb{R}_+$,
    \begin{align*}
        \cost_z(P, C, \Gamma) \in (1 \pm \epsilon) \cdot \cost_z(S, C, \Gamma).
    \end{align*}
\end{definition}

We make an observation in Fact~\ref{fact:assign2kzc} that
an assignment-preserving coreset is as well a coreset for (unconstrained)
clustering.
\begin{fact}
    \label{fact:assign2kzc}
    For $P \subseteq X$, if $S \subseteq P$ is an assignment-preserving
    $\epsilon$-coreset for \kzC on $P$,
    then $S$ is an $\epsilon$-coreset for \kzC on $P$.
\end{fact}
Moreover, this definition of assignment-preserving coresets
generally captures many capacity-constrained clustering problems.
For instance, in capacitated clustering, the goal is to minimize
the \kzC objective subject to the constraint that each center
is assigned by at most a certain number of data points.
Coresets for capacitated clustering have been considered in~\cite{cohen2019fixed}
and our notion of assignment-preserving coresets captures their definition.

\paragraph{Fair clustering.}
In $(\alpha,\beta)$-fair \kzC (\cite{DBLP:conf/nips/Chierichetti0LV17,NEURIPS2019_fc192b0c}), a data set $P$, a collections of groups (not necessary disjoint) $P_1,P_2,...,P_m\subseteq P$ and two constraints vectors $\alpha,\beta\in [0,1]^m$ are given. The objective is to find an assignment $\sigma$ from $P$ to $C$ such that for every group $P_i$ and every center $c\in C$,
$$
\frac{\sigma(P_i,c)}{\sigma(P,c)}\in [\alpha_i,\beta_i].
$$
It has been well known that the requirement of $(\alpha,\beta)$-fair \kzC can be expressed as a collection of assignment constraints~\cite{DBLP:conf/waoa/0001SS19,HJV19,DBLP:conf/icalp/BandyapadhyayFS21,BohmFLMS21}. Following the reduction in \cite{HJV19}, an algorithm that constructs assignment-preserving coresets for \kzC implies coresets algorithm for $(\alpha,\beta)$-fair \kzC (See Section~\ref{sec:fair} for more details).

      \section{New Framework}
\label{sec:framework}

\begin{definition}[Assignment-preserving coresets with additive error] \label{def:assign:add}
Given a data set $P \subseteq X$,
a subset $S \subseteq P$ is called an assignment-preserving
$(\epsilon, A)$-coreset for \kzC on $P$,
if for every $C \subseteq X$ with $|C| \leq k$ and
every assignment constraint $\Gamma : C \to \mathbb{R}_+$,
\begin{align*}
    |\cost_z(P, C, \Gamma) - \cost_z(S, C, \Gamma)|
    \leq \epsilon \cdot (\cost_z(P, C, \Gamma)+ A).
\end{align*}
\end{definition}

The main idea of our new framework (\Cref{thm:reduct_ring}) is to reduce constructing coresets
on general datasets,
to the special case of constructing coresets on datasets that belong to \emph{rings}.
Note that for the rings, we only consider coresets with an additional \emph{additive} error (\Cref{def:assign:add}), which seems to be weaker than
the relative-error coresets that we aim for.
However, by a standard argument, this actually suffices to imply a coreset for the entire dataset without the additive error (see \Cref{sec:additive_err}).

\begin{theorem}
    \label{thm:reduct_ring}
    There is an algorithm that given dataset $P \subseteq X$,
    center $c \in X$, $0 < \epsilon < 1$,
    computes a $2$-partition $\{ W, Z \}$ of $P$ and a weighted point set $S \subseteq P$ of size $2^{O(z\log z)}\cdot\tilde{O}( k \epsilon^{-z})$,
    such that
    \begin{enumerate}
        \item $W$ consists of $2^{O(z\log z)}\cdot\tilde{O}( k \epsilon^{-z})$ rings $\{R_i\}_i$
        where each $R_i \subseteq \ring(c, r_i, 2r_i)$ for some $r_i > 0$, and
        \item $S$ is an assignment-preserving $(\epsilon, \cost_z(P, c))$-coreset for \kzC on $Z$,
    \end{enumerate}
    running in time $\tilde{O}(|P|k)$.
\end{theorem}
Note that the assignment-preserving coreset $S$ for the $Z$ part can be constructed even in general metrics.
Moreover, this assignment-preserving coreset is very general (see Fact~\ref{fact:assign2kzc}),
and it can be used as a coreset for all clustering problems that we consider in this paper.
Hence, in order to obtain a full coreset, it only remains to construct coresets
for $W$, which merely consists of $2^{O(z\log z)}\cdot\tilde{O}( k \epsilon^{-z})$ \emph{rings}.
Therefore, this theorem essentially reduces the coreset construction
for a general data set to ring datasets.
In particular, if one can obtain a coreset (with additive error) of size $T(\epsilon, k, z)$ for each ring,
then one can construct a coreset of size $2^{O(z\log z)}\cdot\tilde{O}( k^2 \epsilon^{-z}) \cdot T(\epsilon, k, z)$ for the entire dataset.

\paragraph{Improved bound for $k = 1$.}
For the special case of \onezC
(noting that when $k=1$ the assignment constraints become trivial and it is equivalent to the un-constrained version),
we have a better argument that yields an improved dependence in $\epsilon$.
\begin{theorem}
    \label{thm:onez}
    There is an algorithm that given dataset $P \subseteq X$,
    center $c \in X$, $0 < \epsilon < 1$,
    computes a $2$-partition $\{ W, Z \}$ of $P$ and a weighted point set $S \subseteq P$ of size $3$,
    such that
    \begin{enumerate}
        \item $W$ consists of $O(\log\frac{z}{\epsilon})$ rings $\{R_i\}_i$
        where each $R_i \subseteq \ring(c, r_i, 2r_i)$ for some $r_i > 0$, and
        \item $S$ is an $(\epsilon, \cost_z(P, c))$-coreset for \onezC on $Z$,
    \end{enumerate}
    running in time $\tilde{O}(|P|k)$.
\end{theorem}

\paragraph{The power of uniform sampling.}
Due to the uniform nature of the ring datasets,
we show in Section~\ref{sec:uniform_vc} that the very simple uniform sampling
already suffices for constructing coresets for \kzC on ring datasets.
This new construction based on uniform sampling
further reduces the construction
of coresets into bounding the \emph{uniform} shattering dimension
of the ball range space induced by the metric space.
The uniform shattering dimension is both easier to analyze,
and wider considered in the literature
than the much more involved weighted shattering dimension used in previous works,
which in turn results in several new and/or improved coreset size bounds.

\subsection{Proof of Theorem~\ref{thm:reduct_ring}}

We provide a sketch of the main algorithm in Algorithm \ref{alg:main_full} and present details of each step in the corresponding paragraph.

\begin{algorithm}[ht]
    \caption{Algorithm Outline for \Cref{thm:reduct_ring}}
    \label{alg:main_full}
    \begin{algorithmic}[1]
        \State set $t \gets \lceil 2+\log \frac{24z k}{\epsilon}\rceil$,
$\err \gets \big(\frac{\epsilon}{6z}\big)^z \cdot \frac{\cost_z(P, c)}{kt}$, and $\mathbb{Z}^*\gets \mathbb{Z}\cup \{-\infty\}$
        \State decompose $P$ into rings $P_i \gets P\cap \ring(c,2^{i-1},2^{i})$ ($i\in \mathbb{Z}^*$) \Comment{call  $P_i$ heavy if $\cost_z(P_i,c)\geq \err$}
        \State mark all heavy rings
        \State merge consecutive unmarked rings to obtain $2^{O(z\log z)}\cdot\tilde{O}( k \epsilon^{-z})$ many groups such that each of the groups has cost at most $\err$, as in Lemma~\ref{lemma:defineQ}
        \State construct a two-points coreset for each group produced in the last step
        \State let $W$ be the union of marked rings, let $Z\gets P\setminus W$, and let $S$ include the union of coresets obtained in the last step
    \end{algorithmic}
\end{algorithm}

\paragraph{Ring decomposition.}
Set $t := \lceil 2+\log \frac{24z k}{\epsilon}\rceil$,
 $\err := \big(\frac{\epsilon}{6z}\big)^z \cdot \frac{\cost_z(P, c)}{kt}$, and   $\mathbb{Z}^*\gets \mathbb{Z}\cup \{-\infty\}$.
If $c\in P$, add $c$ into both $Z$ and $S$ in advance, and let $P\gets P\setminus \{c\}$. Decompose $P$ into rings $\{P_i\mid i\in \mathbb{Z}^*\}$,
where
 for $i\in \mathbb{Z}$,
\begin{align*}
    P_i := P\cap \ring(c,2^{i-1},2^i)=\{p\in P \mid 2^{i-1} <  \dist(p,c)\leq 2^i \}
\end{align*} and if $i=-\infty$, $P_i:=P\cap \{c\}$.

 Since at most $|P|$ rings are non-empty, we can easily compute the above decomposition in near-linear time.

Call a ring $j$ \emph{heavy} if $\cost_z(P_j,c)\geq \mathrm{err}$.
So the number of heavy rings is at most
$\frac{\cost_z(P, c)}{\err}=2^{O(z\log z)}\cdot\tilde{O}( k \epsilon^{-z})$.
We mark all heavy rings. Call a ring \emph{unmarked} if it is not a marked ring.
\paragraph{Defining the partition.}
Now, we define $Z$ as the set of points belong to the unmarked rings,
and define $W$ as the marked rings.
Clearly, $W$ consists of $2^{O(z\log z)}\cdot\tilde{O}( k \epsilon^{-z})$ (marked) rings.
Hence, it remains to construct an $(\epsilon, \cost_z(P, c))$-coreset for $Z$,
the unmarked rings.

\paragraph{Re-grouping unmarked rings.}
Observe that unmarked rings can be grouped into $2^{O(z\log z)}\cdot\tilde{O}( k \epsilon^{-z})$ buckets
of consecutive rings, due to the fact that there are at most $2^{O(z\log z)}\cdot\tilde{O}( k \epsilon^{-z})$
heavy rings.
Denote these buckets as $B_1, B_2, \ldots$,
where each $B_i$ consists of a collection of consecutive unmarked rings.
For every bucket $B_i$, we apply the following Lemma~\ref{lemma:defineQ}
to further group $B_i$ into $\mathcal{G}_i$.
Let $\mathcal{G} := \bigcup_i \mathcal{G}_i$ be the set of all groups of unmarked rings.

\begin{lemma} \label{lemma:defineQ}
    For every bucket $B_i$,
    rings in $B_i$
    can be grouped into $\mathcal{G}_i$
    where each $G \in \mathcal{G}_i$ consists of consecutive unmarked rings,
    such that $\forall G \in \mathcal{G}_i$,
    $\cost_z(G, c) \leq \err$.
    Furthermore,
    $|\bigcup_{i} \mathcal{G}_i| \leq 2^{O(z\log z)}\cdot\tilde{O}( k \epsilon^{-z})$.
\end{lemma}

\begin{proof}
    Fix some $B_i$.
    We start with constructing a grouping $\mathcal{G}_i'$ of $B_i$.
    Initialize $\mathcal{G}_i' := \emptyset$ as the tentative result.
    List points $p$ in $B_i$ in the increasing order of $\dist(p, c)$.
    Starting from the first element in $B_i$,
    greedily pick a maximal subset $G'$ (in order) such that
    $\cost_z(G', c) \leq \err$. Here, we allow $G'$ to include points fractionally.
    Keep on picking such $G'$ and add it into $\mathcal{G}_i'$,
    until all elements in $B_i$ are picked.

    We construct $\mathcal{G}_i$ from $\mathcal{G}_i'$ as follows.
    Observe that each $G' \in \mathcal{G}_i'$
    can \emph{partially} intersect at most two rings from $B_i$.
    Now, examine each $G' \in \mathcal{G}_i'$,
    for every partially intersected ring $R$, designate this entire $R$
    as a new group and include it in $\mathcal{G}_i$,
    remove the partially intersected rings from $G'$,
    and include all rings in the remaining $G'$ as a new group to $\mathcal{G}_i$.
    Eventually, remove the empty or duplicated groups from $\mathcal{G}_i$, if any.

    By construction, $\forall G \in \mathcal{G}_i$,
    either $G$ consists of a single unmarked ring which implies $\cost_z(G, c) \leq \err$, or
    $G$ consists of several consecutive unmarked rings such that
    $\cost_z(G, c) \leq \err$.

    Hence, it remains to bound $|\bigcup_i \mathcal{G}_i|$.
    Since every $G' \in \mathcal{G}_i'$ creates at most $3$ groups in $\mathcal{G}_i$,
    we have $|\mathcal{G}_i| \leq 3 \cdot |\mathcal{G}_i'|$.
    Observe that in every $\mathcal{G}_i'$, there is at most one $G'$ such that
    $\cost_z(G, c) < \err$ and all other $G'$ satisfy $\cost_z(G, c) = \err$.
    Therefore, since there are at most $2^{O(z\log z)}\cdot\tilde{O}( k \epsilon^{-z})$ buckets, over all $i$,
    we have
\begin{align*}
        \left|\bigcup_i \mathcal{G}_i\right| \leq
        2^{O(z\log z)}\cdot\tilde{O}( k \epsilon^{-z})+\frac{\cost(P, c)}{\err}=2^{O(z\log z)}\cdot\tilde{O}( k \epsilon^{-z}).
    \end{align*}
    This finishes the proof of Lemma~\ref{lemma:defineQ}.
\end{proof}

\paragraph{Two-points coresets construction for unmarked groups.}
To construct the coreset $S$ for the unmarked rings, we first construct a two-points
coreset for each group of rings $G_i \in \mathcal{G}$,
and then take the union of them.

    For every group $G_i \in \mathcal{G}$,
    we construct a coreset of only two points using the following steps.
    Let $\pclose_i, \pfar_i \in G_i$ be the closest and furthest points to $c$ (breaking ties consistently), respectively.
    Then for every $p \in G_i$,
    $\dist^z(p, c)$ can be represented by
    $\dist^z(p, c)
    = \lambda_p \cdot \dist^z(\pclose_i, c) + (1 - \lambda_p) \cdot \dist^z(\pfar_i, c)$,
    such that $\lambda_p \in [0, 1]$.
    Then define the weight $w(\pclose_i) := \sum_{p \in G_i} \lambda_p$,
    and similarly the weight of $w(\pfar_i) := \sum_{p \in G_i} (1 - \lambda_p)$.
    Note that $w(\pclose_i) + w(\pfar_i) = w_P(G_i)$ and $\cost_z(\{\pclose_i, \pfar_i\}_w, c) = \cost_z(G_i, c)$,
    where $\{\pclose_i, \pfar_i\}_w$ denotes the two-point coreset.

\paragraph{Size analysis.}
Since the unmarked rings are partitioned into $|\mathcal{G}| \leq 2^{O(z\log z)}\cdot\tilde{O}( k \epsilon^{-z}) $
groups and each of the group creates two coreset points in $S$,
the total number of coreset points is bounded by $|S| \leq 2^{O(z\log z)}\cdot\tilde{O}( k \epsilon^{-z})$.

\paragraph{Error analysis.}
Observe that the requirement of $|Z| = w_S(S)$ follows immediately from the construction, hence we focus on analyzing the coreset cost error.

In the following, we fix a center set $C$ with $|C| = k$,
and an assignment constraint $\Gamma : C \to \mathbb{R}_+$.
We call a ring $P_i$ \emph{important} if there exists $u \in C$, $2^{i-1}< \dist(u,c) \leq 2^i$.
Namely, $P_i$ is important if $C\cap \ring(c,2^{i-1},2^i)\neq\emptyset$.
We color a ring $P_j$ if there exists an important ring $P_i$ such that $|i-j|\leq t$. Namely, for every important ring, we color its $2t$ neighbors and itself.

We call a group $G_i \in \mathcal{G}$ \emph{colored} if it contains any colored ring. We call a group \emph{uncolored} if it is not colored. Colored and uncolored groups have totally different behaviors in producing coreset error.
We bound the coresets errors producing by colored and uncolored groups,
by using Lemma~\ref{lemma:fair:color} and Lemma~\ref{lemma:fair:uncol} respectively.

\begin{lemma}\label{lemma:fair:color}
Let $G_i$ be an unmarked group. Let $\sigma : G_i \times C \to \mathbb{R}_+$
and $\pi : D_i \times C \to \mathbb{R}_+$ both consistent with $\Gamma$.
Then we have
$$
    \left|\cost_z^\sigma(G_i, C)-\cost_z^\pi(D_i, C)\right|
    \leq\epsilon \cdot \cost_z^{\sigma}(G_i,C)+ \frac{\epsilon}{3kt}\cdot\cost_z(P, c).
$$
\end{lemma}

\begin{proof}
As $\sigma$ and $\pi$ are both consistent with $\Gamma$, we know that for every $u\in C$,
$
\sigma(G_i,u)=\pi(D_i,u).
$ Thus we can find a matching $M:G_i\times D_i\to \mathbb{R}_+$, between the mass sent to $u$ by $\sigma$ and $\pi$. Consequently, $M$ satisfies that $\forall x\in G_i$, $\sum_{y\in D_i}M(x,y)=\sigma(x,u)$ and $\forall y\in D_i$, $\sum_{x\in G_i}M(x,y)=\pi(y,u)$.

So by the generalized triangle inequality \cref{lem:gen:tri} we have,
\begin{align*}
&\quad |\cost_z^{\sigma}(G_i,\{u\})-\cost_z^{\pi}(D_i,\{u\})|
\\
&=|\sum_{x\in G_i} \sigma(x,u)\dist^z(x,u)-\sum_{y\in D_i} \pi(y,u)\dist^z(y,u)|\\
&=|\sum_{x\in G_i}\sum_{y\in D_i} M(x,y)\left(\dist^z(x,u)-\dist^z(y,u)\right)|\\
&\leq \sum_{x\in G_i}\sum_{y\in D_i} M(x,y) |\dist^z(x,u)-\dist^z(y,u)|
\\
&\leq \sum_{x\in G_i}\sum_{y\in D_i} M(x,y) \left(\epsilon\cdot \dist^z(x,u)+\big(\frac{3z}{\epsilon}\big)^{z-1}\cdot\dist^z(x,y) \right)
\\
&\leq \epsilon\cdot \cost_z^{\sigma}(G_i,\{u\})+\big(\frac{6z}{\epsilon}\big)^{z-1}\cdot\sum_{x\in G_i}\sum_{y\in D_i}M(x,y)\cdot (\dist^z(x,c)+\dist^z(y,c))
\\
&= \epsilon\cdot \cost_z^{\sigma}(G_i,\{u\})+\big(\frac{6z}{\epsilon}\big)^{z-1}\cdot
\\
&\quad\left(\sum_{x\in G_i}\sigma(x,u)\dist^z(x,c)+\pi(\pclose,u)\dist^z(\pclose,c)+\pi(\pfar,u)\dist^z(\pfar,c)
\right).
\\
\end{align*}

Summing over $u\in C$, we know that,
\begin{align*}
&\quad  \left|\cost_z^\sigma(G_i, C)-\cost_z^\pi(D_i, C)\right|\\
&\leq \sum_{u\in C}|\cost_z^{\sigma}(G_i,\{u\})-\cost_z^{\pi}(D_i,\{u\})|\\
&\leq \epsilon \cdot \cost_z^{\sigma}(G_i,C)+\big(\frac{6z}{\epsilon}\big)^{z-1}\left(\cost_z(G_i,\{c\})+\cost_z(D_i,\{c\})\right)\\
&=\epsilon \cdot \cost_z^{\sigma}(G_i,C)+2\cdot \big(\frac{6z}{\epsilon}\big)^{z-1}\cdot \mathrm{err}\\
&= \epsilon \cdot \cost_z^{\sigma}(G_i,C)+2\cdot \big(\frac{6z}{\epsilon}\big)^{z-1}\cdot  \big(\frac{\epsilon}{6z}\big)^z \cdot \frac{\cost_z(P, c)}{kt}\\
&\leq\epsilon \cdot \cost_z^{\sigma}(G_i,C)+ \frac{\epsilon}{3kt}\cdot\cost_z(P, c).
\end{align*}

\end{proof}

Note that Lemma~\ref{lemma:fair:color} works for both colored and uncolored groups. However, the number of uncolored groups can be much larger than the number of colored groups. Thus we must use a refined analysis Lemma~\ref{lemma:fair:uncol} to bound the error of these uncolored groups.

\begin{lemma} \label{lemma:fair:uncol}
Let $G_i \in \mathcal{G}$ denote an unmarked and uncolored group.
Let $D_i := \{\pclose_i, \pfar_i\}_w$ be the two-point coreset constructed for $G_i$.
The following holds.
\begin{enumerate}
\item For every $\sigma : G_i \times C \to \mathbb{R}_+$ consistent with
    $\Gamma$ (see Definition~\ref{def:assign} for the relevant definitions),
    there exists $\pi : D_i \times C \to \mathbb{R}_+$ consistent with $\Gamma$,
    such that $\cost_z^\pi(D_i, C) \leq (1+\epsilon) \cdot \cost_z^\sigma(G_i, C)$.
    \label{item:fair_one}
    \item For $\pi : D_i \times C \to \mathbb{R}_+$ consistent with $\Gamma$,
    there exists $\sigma : G_i \times C \to \mathbb{R}_+$ consistent with $\Gamma$,
    such that $\cost_z^\sigma(G_i,C) \leq (1+\epsilon)\cost_z^\pi(D_i, C)$.
    \label{item:fair_two}
\end{enumerate}
In particular, let $\sigma$ and $\pi$ denote the optimal assignment consistent with $\Gamma$ from $G_i$ and $D_i$ to $C$ respectively, then $$
|\cost_z^{\sigma}(G_i,C)-\cost_z^{\pi}(D_i,C)|\leq O(\epsilon)\cdot \cost_z^{\sigma}(G_i,C).
$$
\end{lemma}

\begin{proof}
Divide $C$ into $\Cclose := \{u \in C \mid \dist(u,c) < \dist(\pclose_i, c)\}$
and $\Cfar := C \setminus \Cclose$.
Recall that $G_i$ is uncolored. We need the following inequalities to characterize the distances between $G_i$ and $\Cfar$ or $\Cclose$.

\begin{lemma}\label{lemma:uncol:distance}
The following inequalities hold.

\begin{enumerate}
    \item For every $u\in \Cclose$, $ \dist(\pclose_i, c) > \frac{9z}{\epsilon}\dist(u, c)$.
    \label{uncol:item:close}
    \item For every $u\in \Cfar$,
    $\dist(\pfar_i, c)
    <\frac{\epsilon}{24z} \cdot \dist(u,c)$.
    \label{uncol:item:far}
    \item For every $x\in G_i$ and $u\in \Cfar$, $\dist(\pfar_i,c)< \frac{\epsilon}{12z}\cdot\dist(x,u)$. \label{uncol:inq1}

    \item For every $x\in G_i$ and $u\in \Cfar$, $\dist^z(x,u)\in (1\pm \frac{\epsilon}{6})\cdot \dist^z(u,c)$. \label{uncol:inq3}
    \item For every $x\in G_i$ and $u\in \Cclose$, $\dist^z(x,u)\in (1\pm \frac{\epsilon}{6})\cdot \dist^z(x,c)$.
    \label{uncol:inq4}
\end{enumerate}
\end{lemma}

\begin{proof}
For item (\ref{uncol:item:close}) and (\ref{uncol:item:far}), assume $u$ is in $P_j$, so $P_j$ is important. As $G_i$ is not colored, by definition, we know that for every $P_{i'}\subseteq G_i$, $|i'-j|>t$. So if $u\in \Cclose$, we have
$$
\dist(\pclose_i,c)\geq 2^{t-1}\cdot \dist(\pfar_j,c)> \frac{9z}{\epsilon} \cdot \dist(u,c).
$$
On the other hand, if $u\in \Cfar$, we have
$$
\dist(\pfar_i,c)\leq 2^{-t+1}\cdot \dist(\pclose_j,c)\leq \frac{\epsilon}{24z}\cdot \dist(u,c).
$$
\\

Now we prove item (\ref{uncol:inq1}).
By item (\ref{uncol:item:far}), $\max_{x\in G_i} \dist(x,c) = \dist(\pfar_i, c)
    <\frac{\epsilon}{24z} \cdot \dist(u,c)$, by triangle inequality we know that for every $x\in G_i$, $$\dist(x,u)\geq \dist(u,c)-\dist(x,c)>\big(\frac{24z}{\epsilon}-1\big)\dist(\pfar_i,c)>\frac{12z}{\epsilon}\cdot \dist(\pfar_i,c).$$ So we know that for every $u\in \Cfar$ and $x\in G_i$, $\dist(\pfar_i,c)<\frac{\epsilon}{12z} \cdot \dist(x,u)$.
\\

Now we prove item (\ref{uncol:inq3}). By the triangle inequality and item (\ref{uncol:item:far}) we know that,
$$
|\dist(x,u)-\dist(u,c)|\leq \dist(x,c)\leq \dist(\pfar_i,c)<\frac{\epsilon}{24z}\cdot \dist(u,c).
$$ Thus $\dist(x,u)\in \big(1\pm \frac{\epsilon}{24z}\big)\cdot \dist(u,c)$. Thus $\dist^z(x,u)\in \big(1\pm \frac{\epsilon}{6}\big)\cdot \dist^z(u,c)$.
\\

Now we prove item (\ref{uncol:inq4}). By the triangle inequality and item (\ref{uncol:item:close}), we know that
$$
|\dist(x,u)-\dist(x,c)|\leq \dist(u,c)<\frac{\epsilon}{9z}\cdot \dist(\pclose_i,c)\leq \frac{\epsilon}{9z}\cdot\dist(x,c).
$$
Thus $\dist(x,u)\in \big(1\pm \frac{\epsilon}{9z})\cdot \dist(x,c)$. Thus $\dist^z(x,u)\in (1\pm \frac{\epsilon}{6})\cdot \dist^z(x,c)$.
This finishes the proof of Lemma~\ref{lemma:uncol:distance}.
\end{proof}

We return to the proof of Lemma \ref{lemma:fair:uncol}, we prove item (\ref{item:fair_one}) and item (\ref{item:fair_two}) separately.

\paragraph{Proof of item (\ref{item:fair_one}).}
Recall that for every $x \in G_i$, there exists a unique $\lambda_x$
such that $\dist^z(x, c) = \lambda_x \dist^z(\pclose_i, c) + (1 - \lambda_x) \dist^z(\pfar_i,c)$,
and we have set $w(\pclose_i) = \sum_{x\in G_i} \lambda_x$
and $w(\pfar_i)=\sum_{x\in G_i} (1-\lambda_x)$.

Consider some $\sigma : G_i \times C \to \mathbb{R}_+$ consistent with $\Gamma$.
We define $\pi : D_i \times C \to \mathbb{R}_+$,
such that for every $u \in C$, $\pi(\pclose_i, u) = \sum_{x\in G_i} \lambda_x\sigma(x,u)$
and $\pi(\pfar_i, u) = \sum_{x\in G_i} (1-\lambda_x)\sigma(x, u)$.
We note that for every $u\in C$,
$$
    \pi(D_i, u) = \sigma(G_i, u).
$$
Hence $\pi$ is consistent with $\Gamma$.
It remains to prove that for every $u\in C$,
$$
    \pi(\pclose_i, u) \dist^z(\pclose_i, u) + \pi(\pfar_i, u) \dist^z(\pfar_i, u)
    \leq (1+\epsilon) \sum_{x\in G_i} \sigma(x, u) \cdot \dist^z(x, u).
$$
\begin{itemize}
    \item If $u \in \Cfar$, by item (\ref{uncol:inq1}) of Lemma \ref{lemma:uncol:distance}, we know that for every $u\in C_{far}$ and $x\in G_i$, $\dist(\pfar_i,c)< \frac{\epsilon}{12z} \cdot \dist(x,u)$. So by the generalized triangle inequality Lemma \ref{lem:gen:tri}  we have
    \begin{align*}
        &\quad \pi(\pclose_i, u) \dist^z(\pclose_i, u) + \pi(\pfar_i, u) \dist^z(\pfar_i, u) \\
        &= \sum_{x\in G_i} \left( \lambda_x \sigma(x, u) \dist^z(\pclose_i, u)
            +\left(1-\lambda_x\right)\sigma(x, u) \dist^z(\pfar_i, u) \right)\\
        &= \sum_{x\in G_i} \sigma(x, u) \left( \lambda_x \dist^z(\pclose_i, u)
            +\left( 1-\lambda_x \right) \dist^z(\pfar_i, u) \right)\\
            &\leq
          \sum_{x\in G_i} \sigma(x, u) \bigg( \lambda_x \big((1+\frac{\epsilon}{2})\dist^z(x,u)+(6z/\epsilon)^{z-1}\dist^z(\pclose_i,x) \big)\\
 &\quad+\left( 1-\lambda_x \right)\big((1+\frac{\epsilon}{2})\dist^z(x,u)+(6z/\epsilon)^{z-1}\dist^z(\pfar_i,x) \big) \bigg)\\
            &\leq \big(1+\frac{\epsilon}{2}\big)\left(\sum_{x\in G_i} \sigma(x, u)\cdot \dist^z(x, u) \right)+\big(\frac{6z}{\epsilon}\big)^{z-1}\cdot\sum_{x\in G_i}\sigma(x,u)\cdot \big(2\dist(\pfar_i,c)\big)^z
            \\
        &\leq \big(1+\frac{\epsilon}{2}\big)\left(\sum_{x\in G_i} \sigma(x, u)\cdot \dist^z(x, u) \right)+\big(\frac{6z}{\epsilon}\big)^{z-1}\cdot\sum_{x\in G_i}\sigma(x,u)\cdot \big(\frac{\epsilon}{6z}\cdot\dist(x,u)\big)^z
            \\       &\leq\big(1+\frac{\epsilon}{2}\big) \left(\sum_{x\in G_i} \sigma(x, u)\cdot \dist^z(x, u) \right)
            + \frac{\epsilon}{2}\cdot \sum_{x\in G_i} \sigma(x, u)\cdot \dist^z(x, u)\\
        &\leq (1+\epsilon) \sum_{x\in G_i} \sigma(x, u)\cdot \dist^z(x, u).
    \end{align*}

    \item If $u\in \Cclose$, we first observe that by construction,
    \begin{align*}
        &\quad \pi(\pclose_i, u) \dist^z(\pclose_i, c) + \pi(\pfar_i, u) \dist^z(\pfar_i, c)\\
        &=\sum_{x\in G_i} \left(\lambda_x\sigma(x,u)\dist^z(\pclose,c)+(1-\lambda_x)\sigma(x,u)\dist^z(\pfar,c)\right)
        \\
        &=\sum_{x\in G_i} \sigma(x, u) \dist^z(x, c).\\
        \end{align*}
   It remains to show that replacing $c$ with $u$ produces affordable error. By item (\ref{uncol:item:close}) of Lemma \ref{lemma:uncol:distance}, we know that $\min_{x\in G_i} \dist(x,c) = \dist(\pclose_i, c) > \frac{9z}{\epsilon}\dist(u, c)$. So we know that for every $x\in G_i$, $\dist(u,c)<\frac{\epsilon}{9z}\cdot \dist(x,u)$.

   By generalized triangle inequality Lemma \ref{lem:gen:tri} we have,
   \begin{align*}
       &\quad \pi(\pclose_i, u) \dist^z(\pclose_i, u) + \pi(\pfar_i, u) \dist^z(\pfar_i, u) \\
       &\leq \big(1+\frac{\epsilon}{3}\big)\big(\pi(\pclose_i, u) \dist^z(\pclose_i, c)
            + \pi(\pfar_i, u) \dist^z(\pfar_i, c)\big) +(\frac{9z}{\epsilon})^{z-1}\cdot \pi(D_i,u) \cdot \dist^z(u, c)\\
       &= \big(1+\frac{\epsilon}{3}\big)\sum_{x\in G_i} \sigma(x,u) \dist^z(x, c) + (\frac{9z}{\epsilon})^{z-1}\cdot \sigma(G_i, u) \cdot \dist^z(u, c)\\
       &\leq \big(1+\frac{2\epsilon}{3}\big)\sum_{x\in G_i} \sigma(x, u) \dist^z(x, u) + 3\cdot (\frac{9z}{\epsilon})^{z-1}\cdot\sigma(G_i, u)\cdot \dist^z(u, c)\\
       &\leq \big(1+\frac{2\epsilon}{3}\big)\sum_{x\in G_i} \sigma(x, u) \dist^z(x, u) + 3\cdot (\frac{9z}{\epsilon})^{z-1}\cdot \sum_{x\in G_i} \sigma(x, u)\cdot \big(\frac{\epsilon}{9z}\cdot\dist(x, u)\big)^z\\
       &\leq \left(1+\epsilon\right) \sum_{x\in G_i} \sigma(x, u) \dist^z(x,u).
   \end{align*}
\end{itemize}

\paragraph{Proof of item (\ref{item:fair_two}).}
Consider $\pi : D_i \times C \to \mathbb{R}_+$ consistent with $\Gamma$.
We need to construct $\sigma : G_i \times C \to \mathbb{R}_+$
so that $\cost_z^\sigma(G_i, C)\leq (1+\epsilon) \cost_z^\pi(D_i, C)$.
We find such $\sigma$ by considering the following linear program,
\begin{equation*}
\begin{array}{ll@{}ll}
\text{minimize}  & \displaystyle \sum_{x\in G_i} \sum_{u\in C} \sigma(x, u) \cdot \dist^z(x, u)&\\
\text{subject to}&\displaystyle  \sigma(x, u)\geq 0 &\;& \forall x\in G_i, u \in C,\\
                & \displaystyle \sum_{u \in C} \sigma(x, u)=1 &\;&\forall x\in G_i,\\
                 & \displaystyle \sum_{x\in G_i} \sigma(x, u) = \Gamma(u) &\;&\forall u\in C
\end{array}
\end{equation*}

The above linear programming is clearly a feasible min-cost flow problem as
there must exist $\sigma : G_i \times C \to \mathbb{R}_+$ consistent with $\Gamma$.
Let $\sigma$ denote the optimal solution of the LP.
It suffices to show $\cost_z^\sigma(G_i, C) \leq \left(1+\epsilon\right)
\cost_z^\pi(D_i, C)$.
We need the following Lemma~\ref{lemma:fair:XitoPi}.

\begin{lemma} \label{lemma:fair:XitoPi}
The following inequalities for $\sigma$ hold.
\begin{enumerate}
    \item For every $u\in \Cfar$,
    \begin{align*}
        \sum_{x\in G_i} \sigma(x, u)\dist^z(x, u)
        \leq \left(1+\frac{\epsilon}{2}\right)
            \left(\pi(\pclose_i, u) \dist^z(\pclose_i, u)
            +\pi(\pfar_i, u) \dist^z(\pfar_i, u) \right).
    \end{align*} \label{XitoPi:item1}
    \item We have the following for $\Cclose$,
    \begin{align*}
        &\quad
        \sum_{x\in G_i}\sigma(x,\Cclose) \dist^z(x,c) \\
        &\leq \pi(\pclose_i, \Cclose) \dist^z(\pclose_i, c) + \pi(\pfar_i, \Cclose) \dist^z(\pfar_i, c)
        + \frac{\epsilon}{3}\cdot \cost_z^\pi(D_i, C).
    \end{align*} \label{XitoPi:item2}
    \item $\cost_z^\sigma(G_i,\Cclose)\leq \big(1+\frac{\epsilon}{2}\big)\cost_z^\pi(D_i,\Cclose)
    + \frac{\epsilon}{2}\cdot \cost_z^\pi(D_i,C)$.
    \label{XitoPi:item3}
\end{enumerate}
\end{lemma}

\begin{proof}
For item (\ref{XitoPi:item1}), we note that $\sigma(G_i, u)= \pi(D_i, u)$
for every $u \in C$. By item (\ref{uncol:inq3}) of Lemma \ref{lemma:uncol:distance}, we know that for every $x\in G_i$ and $u\in C_{far}$, $\dist^z(x,u)\in (1\pm\frac{\epsilon}{6}) \dist^z(u,c)$.

So we have
\begin{eqnarray*}
    \sum_{x\in G_i} \sigma(x, u)\cdot \dist^z(x, u)
    &\leq & \left(1+\frac{\epsilon}{6}\right) \sigma(G_i, u) \dist^z(u, c)\\
    &=&\left(1+\frac{\epsilon}{6}\right) \pi(D_i, u) \dist^z(u, c)\\
    &\leq & \frac{1+\frac{\epsilon}{6}}{1-\frac{\epsilon}{6}}\cdot \left(\pi(\pclose_i, u)d(\pclose_i, u)
        +\pi(\pfar_i, u) \dist(\pfar_i, u)\right)\\
         &\leq & \left(1+\frac{\epsilon}{2}\right)\cdot \left(\pi(\pclose_i, u)d(\pclose_i, u)
        +\pi(\pfar_i, u) \dist(\pfar_i, u)\right).
\end{eqnarray*}
For item (\ref{XitoPi:item2}), we first note that
\begin{eqnarray*}
    \sum_{x\in G_i}\sigma(x,\Cclose) \dist^z(x,c)
    &=& \sum_{x\in G_i}\sigma(x,\Cclose) \left(  \lambda_x \dist^z(\pclose_i, c)
        + \left(1-\lambda_x\right) \dist^z(\pfar_i,c)  \right)\\
    &=& a_1\cdot \dist^z(\pclose_i,c) + a_2\cdot \dist^z(\pfar_i,c)
\end{eqnarray*}
where $a_1 := \sum_{x\in G_i}\lambda_x \sigma(x,\Cclose)$
and $a_2 := \sum_{x\in G_i}\left(1-\lambda_x\right)\sigma(x,\Cclose)$.
We observe that $$a_1 + a_2 =\sigma(G_i,\Cclose)= \Gamma(\Cclose)=\pi(D_i,\Cclose)
=\pi(\pclose_i, \Cclose) + \pi(\pfar_i, \Cclose)$$ which implies
$$
|a_1-\pi(\pclose_i,\Cclose)|=|a_2-\pi(\pfar_i,\Cclose)|.
$$
So we have,
\begin{align*}
    &\quad \sum_{x\in G_i}\sigma(x, \Cclose) \dist^z(x,c)
        - \left(\pi(\pclose_i, \Cclose) \dist^z(\pclose_i, c)
        + \pi(\pfar_i, \Cclose) \dist^z(\pfar_i, c)\right)\\
    &= a_1\cdot \dist^z(\pclose_i,c) + a_2\cdot \dist^z(\pfar_i,c)-\left(\pi(\pclose_i, \Cclose) \dist^z(\pclose_i, c)
        + \pi(\pfar_i, \Cclose) \dist^z(\pfar_i, c)\right)\\
    &\leq |a_1 - \pi(\pclose_i, \Cclose)| \dist^z(\pclose_i, c)
        + |a_2 - \pi(\pfar_i, \Cclose)| \dist^z(\pfar_i,c)\\
    &= |a_1 - \pi(\pclose_i, \Cclose)|\cdot \left(\dist^z(\pclose_i,c)+\dist^z(\pfar_i,c)\right)\\
    &\leq 2|a_1 - \pi(\pclose_i, \Cclose)|\cdot \dist^z(\pfar_i, c)\\
    &\leq  2\left(|a_1-w(\pclose)|+|w(\pclose)-\pi(\pclose_i,\Cclose)|\right)\cdot\dist^z(\pfar_i,c)\\
     &=  2\left(\bigg|\sum_{x\in G_i}\lambda_x\sigma(x,\Cclose)-\sum_{x\in G_i}\lambda_x\bigg|+\pi(\pclose_i,\Cfar)\right)\cdot\dist^z(\pfar_i,c)\\
          &=  2\big(\sum_{x\in G_i}\lambda_x\sigma(x,\Cfar)+\pi(\pclose_i,\Cfar)\big)\cdot\dist^z(\pfar_i,c)\\
    &\leq  4 \Gamma (\Cfar) \cdot \dist^z(\pfar_i, c)\\
    &\leq  4 \Gamma (\Cfar)\cdot \left(\frac{\epsilon}{12z}\right)^z\cdot \min \left( \dist^z(\pclose_i, \Cfar), \dist^z(\pfar_i, \Cfar)\right)
     \\
    &\leq \frac{\epsilon}{3}\cdot \cost_z^\pi(D_i, C).
\end{align*}
Now we are ready to prove item (\ref{XitoPi:item3}).
This is a simple corollary of the second item.
Actually, by item (\ref{uncol:inq4}) of Lemma \ref{lemma:uncol:distance}, we know that for every $x\in G_i$ and $u\in \Cclose$,
$$
    \dist^z(x, u)\in \left(1\pm \frac{\epsilon}{6}\right)\cdot  \dist^z(x,c).
$$
Thus we have,
$$
 \cost_z^\sigma(G_i,\Cclose)\leq \big(1+\frac{\epsilon}{6}\big) \cdot \sum_{x\in G_i}\sigma(x,\Cclose) \dist^z(x,c)
$$
and
$$
    \pi(\pclose_i, \Cclose) \dist(\pclose_i,c) + \pi(\pfar_i, \Cclose) \dist(\pfar_i, c)\leq \frac{1}{1-\epsilon/6}
        \cdot\cost_z^\pi(D_i,\Cclose).
$$

Thus we have
\begin{align*}
    &\quad\cost_z^{\sigma}(G_i,\Cclose)
    \\
    &\leq \big(1+\frac{\epsilon}{6}\big) \sum_{x\in G_i}\sigma(x,\Cclose) \dist^z(x,c)
    \\
    &\leq \big(1+\frac{\epsilon}{6}\big)\left(\pi(\pclose_i, \Cclose) \dist^z(\pclose_i, c) + \pi(\pfar_i, \Cclose) \dist^z(\pfar_i, c)
        + \frac{\epsilon}{3}\cdot \cost_z^\pi(D_i, C)\right)
    \\
    &\leq
    \frac{1+\epsilon/6}{1-\epsilon/6}\cdot \cost_z^{\pi}(D_i,\Cclose)+\frac{\epsilon}{3}\cdot \big(1+\frac{\epsilon}{6}\big)\cdot\cost_z^{\pi}(D_i,C)
    \\
    &\leq \big(1+\frac{\epsilon}{2})\cdot\cost_z^{\pi}(D_i,\Cclose)+\frac{\epsilon}{2}\cdot \cost_z^{\pi}(D_i,C).
\end{align*}

Thus we have proved Lemma \ref{lemma:fair:XitoPi}
\end{proof}

Now we are ready to prove Lemma \ref{lemma:fair:uncol}. Item (\ref{XitoPi:item1}) and Item (\ref{XitoPi:item3}) of Lemma \ref{lemma:fair:XitoPi} imply
$\cost_z^\sigma(G_i,\Cfar) \leq \left(1+\frac{\epsilon}{2}\right) \cost_z^\pi(D_i,\Cfar)$ and $\cost_z^\sigma(G_i,\Cclose)\leq \big(1+\frac{\epsilon}{2}\big)\cdot \cost_z^\pi(D_i,\Cclose)
    + \frac{\epsilon}{2}\cdot \cost_z^\pi(D_i,C)$.
Combining with them, we have
\begin{eqnarray*}
    \cost_z^\sigma(G_i, C)
    &=&\cost_z^\sigma(G_i,\Cclose) + \cost_z^\sigma(G_i,\Cfar)\\
    &\leq& \big(1+\frac{\epsilon}{2}\big)\cdot \cost_z^\pi(D_i, \Cclose) + \big(1+\frac{\epsilon}{2}\big)\cdot\cost_z^\pi(D_i, \Cfar)
        + \frac{\epsilon}{2}\cdot \cost_z^\pi(D_i,C)\\
    &= & \left(1+\epsilon\right)\cdot \cost_z^\pi(D_i,C).
\end{eqnarray*}

Thus we have proved Lemma \ref{lemma:fair:uncol}.
\end{proof}

\paragraph{Concluding the error analysis.}
Now we are ready to finish the error analysis for the coreset $S$ on the unmarked groups $Z$. It can be simply done by combing Lemma~\ref{lemma:fair:color} and Lemma~\ref{lemma:fair:uncol}. Recall that by construction there are at most $k(2t+1)$ colored rings and every colored group contains at least one colored rings. Since groups contain disjoint rings, we know that there are at most $k(2t+1)$ many colored groups. Let $\sigma$ denote the optimal assignment consistent with $\Gamma$ from $Z$ to $C$. Recall that $Z=\bigcup_{i:G_i\in\mathcal{G}} G_i$. Let $\Gamma_i$ denote the assignment constraints such that $\forall u\in C,\Gamma_i(u)=\sum_{x\in G_i} \sigma(x,u)$. Let $\pi_i$ denote the optimal assignment from $D_i$ to $C$ consistent with $\Gamma_i$. By Lemma~\ref{lemma:fair:color} and Lemma~\ref{lemma:fair:uncol} we have,

\begin{align*}
    &\quad \cost_z(S,C,\Gamma)
    \\
    &\leq \sum_{i:G_i\in \mathcal{G}} \cost_z(D_i,C,\Gamma_i)\\
    &= \sum_{i:G_i\in \mathcal{G}} \cost_z^{\pi_i}(D_i,C)\\
    &\leq (1+\epsilon)\cdot  \sum_{i:G_i\in \mathcal{G}} \cost_z^{\sigma}(G_i,C)+k(2t+1)\cdot \frac{\epsilon}{3kt}\cdot \cost_z(P,c)
    \\
    &\leq (1+\epsilon)\cdot \cost_z(Z,C,\Gamma)+\epsilon\cdot\cost_z(P,c).
\end{align*}

Similarly, we have that
$$
\cost_z(Z,C,\Gamma)\leq (1+\epsilon)\cdot \cost_z(S,C,\Gamma)+\epsilon\cdot\cost_z(P,c)
$$
and conclude that
$$
|\cost_z(Z,C,\Gamma)-\cost_z(S,C,\Gamma)|\leq O(\epsilon)\cdot \cost_z(Z,C,\Gamma).
$$
It remains to scale $\epsilon$.
\qed

\subsection{Proof of Theorem~\ref{thm:onez}}

\begin{theorem}[Restatement of Theorem~\ref{thm:onez}]
    There is an algorithm that given dataset $P \subseteq X$,
    center $c \in X$, $0 < \epsilon < 1$,
    computes a $2$-partition $\{ W, Z \}$ of $P$ and a weighted point set $S \subseteq P$ of size $3$,
    such that
    \begin{enumerate}
        \item $W$ consists of $O(\log\frac{z}{\epsilon})$ rings $\{R_i\}_i$
        where each $R_i \subseteq \ring(c, r_i, 2r_i)$ for some $r_i > 0$, and
        \item $S$ is an $(\epsilon, \cost_z(P, c))$-coreset for \onezC on $Z$,
    \end{enumerate}
    running in time $\tilde{O}(|P|k)$.
\end{theorem}

\paragraph{Ring decomposition.}
Let $r := \left(\frac{\cost_z(P, c)}{|P|}\right)^{1/z}$ denote the average cost of $P$.
We decompose $P$ into $3$ groups.
\begin{itemize}
    \item $\Pclose = \{p\in P \mid d(p,c)<\frac{\epsilon }{6z}\cdot r\}$.
    \item $\Pfar = \{p\in P \mid d(p,c)>\frac{120z}{\epsilon^2} r\}$.
    \item $\Pmain = P\setminus(\Pclose \bigcup \Pfar)$.
\end{itemize}
Define $W := \Pmain$, and it is clear that $W$ can be covered
by a union of $O(\log \frac{z}{\epsilon})$ rings of the form
$\ring(c, a, 2a)$ for some $a \geq 0$.
Define $Z := P \setminus W$, then $Z = \Pclose \cup \Pfar$.
It remains to define an $(\epsilon, \cost_z(P, c))$-coreset $S$ for $Z$.

\paragraph{Constructing coreset $S$.}
Recall that $Z = \Pclose \cup \Pfar$, so we construct coresets for $\Pclose$
and $\Pfar$ separately, and take the union of them.

\begin{itemize}
    \item For $\Pclose$, we add to $S$ a single coreset point $c$
    with weight $w(c) := |\Pclose|$. Note that if one insist looking for a subset of $P$ as coreset, one can replace $c$ with the closet point  $c_{\mathrm{min}}\in \Pclose$ to $c$. It only remains to scale $\epsilon$.
    \item For $\Pfar$, let $\pfar,\pclose\in \Pfar$ denote the further and closest point to $c$. For every $x\in \Pfar$ there is a unique $\lambda_x$ such that $\dist^z(x,c)=\lambda_x\cdot \dist^z(\pclose,c)+(1-\lambda_x)\cdot\dist^z(\pfar,c)$. We add $\pfar$ and $\pclose$ to $S$ and set the weight as $w(\pclose)=\sum_{x\in \Pfar} \lambda_x$ and $w(\pfar)=\sum_{x\in \Pfar} (1-\lambda_x)$. Note that $w(\pclose)+w(\pfar)=|\Pfar|$ and $\cost_z(\{\pclose,\pfar\}_w,c)=\cost_z(\Pfar,c)$.
\end{itemize}
Clearly $|S| = 3$ and we argue the error bound in the following.

\paragraph{Error analysis.}

Fix an arbitrary center $s\in X$.
For $\Pclose$, since we move all points in $\Pclose$ to $c$, by generalized triangle inequality Lemma~\ref{lem:gen:tri},
it incurs at most $$\frac{\epsilon}{2} \cdot \cost_z(\Pclose,s)+\big(\frac{6z}{\epsilon}\big)^{z-1}\cdot |\Pclose|\cdot\big (\frac{\epsilon r}{6z}\big)^z\leq \frac{\epsilon}{2}\cdot \left(\cost_z(P,s)+\cost_z(P,c)\right)$$ error.

So it suffices to prove that $$|\cost_z(\Pfar,s)-\cost_z(\{\pclose,\pfar\}_w,c)|\leq \frac{\epsilon}{2}\cdot \left(\cost_z(P,s)+\cost_z(P,c)\right).$$

In the following, we do a case-analysis with respect depending on whether
$ \dist(s, c)>\frac{5r}{\epsilon}$.

If $\dist(s, c) > \frac{5r}{\epsilon}$, we first observe that
$\cost_z(P, s) \geq \frac{2^{z+1}}{ \epsilon}\cdot \cost_z(P, c)$.
To see this, we note that in $P$, at least $|P|/2$ points are $r/2$ close to $c$ and by triangle inequality,
they are at least $(\frac{5}{\epsilon}-\frac{1}{2})\cdot r\geq \frac{4}{\epsilon}\cdot r$ far from $s$.
So $\cost_z(P, s) > \left(\frac{4r}{\epsilon}\right)^z\cdot \frac{|P|}{2}
\geq \frac{2^{z+1}}{\epsilon}\cdot \cost_z(P,c)$.
By generalized triangle inequality Lemma~\ref{lem:gen:tri} and the fact that $|\Pfar|=w(\pclose)+w(\pfar)$ we have,
\begin{align*}
|\cost_z(\Pfar, s) - \cost_z(\{\pclose,\pfar\}_w, s)|
&\leq 2^{z-1}\left(\cost_z(\Pfar, c)
+\cost_z(\{\pclose,\pfar\}_w, c) \right) \\
&= 2^z\cost_z(P,c)
\\
&\leq \frac{\epsilon}{2}\cdot \cost_z(P,s).
\end{align*}

Now consider the other case, which is $\dist(s, c)\leq \frac{5r}{\epsilon}$.
By construction we know that $$\min_{x\in \Pfar}\dist(x, c) >\frac{120z}{\epsilon^2}\cdot r\geq \frac{24z}{\epsilon}\cdot \dist(s,c).$$ By generalized triangle inequality Lemma~\ref{lem:gen:tri} we know that
\begin{align*}
&\quad |\cost_z(\Pfar, s) - \cost_z(\Pfar, c)| \\
&\leq  \frac{\epsilon}{8}\cdot \cost_z(\Pfar,s)+\big(\frac{24z}{\epsilon}\big)^{z-1}\cdot |\Pfar|\cdot \dist^z(s,c)
\\
&\leq \frac{\epsilon}{8}\cdot \cost_z(\Pfar,s) + \big(\frac{24z}{\epsilon}\big)^{z-1}\cdot |\Pfar|\cdot  \left(\frac{\epsilon\cdot \min_{x\in \Pfar} \dist(x,c)}{24z}\right)^z\\
&= \frac{\epsilon}{8}\cdot \cost_z(\Pfar,s)+\frac{\epsilon}{8} \cdot \cost_z(P,c).
\end{align*}
Similarly,
\begin{align*}
&\quad |\cost_z(\{\pclose,\pfar\}_w, s) - \cost_z(\{\pclose,\pfar\}_w, c)| \\
&\leq  \frac{\epsilon}{8}\cdot \cost_z(\{\pclose,\pfar\}_w,s)+\big(\frac{24z}{\epsilon}\big)^{z-1}\cdot \big(w(\pfar)+w(\pclose)\big)\cdot \dist^z(s,c)
\\
&\leq \frac{\epsilon}{8}\cdot \cost_z(\{\pclose,\pfar\}_w,s) + \big(\frac{24z}{\epsilon}\big)^{z-1}\cdot\big(w(\pfar)+w(\pclose)\big)\cdot  \left(\frac{\epsilon\cdot  \dist(\pclose,c)}{24z}\right)^z\\
&\leq \frac{\epsilon}{8}\cdot \cost_z(\{\pclose,\pfar\}_w,s) + \frac{\epsilon}{8}\cdot\cost_z(\{\pclose,\pfar\}_w,c)\\
&= \frac{\epsilon}{8}\cdot \cost_z(\{\pclose,\pfar\}_w,s) + \frac{\epsilon}{8}\cdot\cost_z(\Pfar,c)\\
&\leq \frac{\epsilon}{8}\cdot \cost_z(\{\pclose,\pfar\}_w,s)+\frac{\epsilon}{8}\cdot \cost_z(P,c).
\end{align*}

Recall that $\cost_z(\Pfar, c) = \cost_z(\{\pclose,\pfar\}_w, c)$ by construction.
Combining the above two inequalities we can show that, $$
|\cost_z(\{\pfar,\pclose\}_w, s) - \cost_z(\Pfar, s)| \leq  \frac{\epsilon}{2} \cdot\left(\cost_z(P,s)+ \cost_z(P, c)\right).
$$

\qed

     \newcommand{\Tclose}{\ensuremath{T_{\mathrm{close}}\xspace}}
\newcommand{\Tfar}{\ensuremath{T_{\mathrm{far}}\xspace}}

\section{Assignment-preserving Coresets for Rings in $\mathbb{R}^d$}
\label{sec:fair}
In this section, we show how to construct assignment-preserving coresets for \kMedian. For simplicity, throughout this section, we use $\cost(\cdot)$ to represent $\cost_1(\cdot)$.

\begin{theorem}\label{thm:ring_fair}
Let $c\in \mathbb{R}^d$, $r>0$, and $P\subseteq \ring(c,r,2r)$ be a dataset with $|P|=n$.
Let $D \subseteq P$ be a uniform sample of size
$m=\tilde{O}(\frac{k}{\epsilon^5}\log \delta^{-1})$ and re-weight $D$ such that $\forall x\in D, w_D(x) := \frac{n}{m}$.
Then with probability at least $1-\delta $, $D$ is an assignment-preserving $(\epsilon,nr)$-coreset for \kMedian.
\end{theorem}

We can assume the input dimension $d = \tilde{O}(\epsilon^{-2}\log k)$,
by applying the iterative size reduction technique introduced in recent paper~\cite{BJKW21}
which is based on a terminal version of Johnson-Lindenstrauss Lemma~\cite{NN19}.\footnote{Strictly speaking, the iterative size reduction technique in \cite{BJKW21} is designed for classical $\epsilon$-coresets instead of our assignment-preserving $(\epsilon,A)$-coresets for \kMedian. The algorithm in \cite{BJKW21} iteratively construct $\epsilon_{i+1}$-coreset on $\epsilon_i$-coreset with carefully chosen $\epsilon_i$'s. Here, since we only require the argument work for a \emph{fixed} $A=nr$, we can apply the reduction in an identical way by iteratively constructing $(\epsilon_{i+1},A)$-coreset on $(\epsilon_i,A)$-coreset with the same set of $\epsilon_i$'s.} Thus it suffices to prove Theorem~\ref{thm:ring_fair} with target coreset size $m = \tilde{O}(\frac{kd}{\epsilon^3})$.

The following lemma shows that
it suffices to bound $| \cost(P,C,\Gamma)-\cost(D,C,\Gamma)|$,
for a $k$-point center set $C$ with assignment constraint $\Gamma$ such that the
total mass of assignment for the ``far'' portion of $C$ is small.

\begin{lemma}\label{lemma:ring_fair_far}
Let $P$ and $D$ be the dataset and coreset in Theorem~\ref{thm:ring_fair}. Let $C\subseteq \mathbb{R}^d,|C|=k$ and $\Gamma : C \to \mathbb{R}_+$ be an assignment constraint such that $\sum_{u\in C}\Gamma(u)=n$.
Let $\Cfar=\{u\in C\mid \dist(u,c)>5kr/\epsilon^2\}$.
If $\Gamma(\Cfar)=\sum_{u\in \Cfar} \Gamma(u)>\epsilon n/k$, $$
\big | \cost(P,C,\Gamma)-\cost(D,C,\Gamma)\big|<\epsilon \cost(P,C,\Gamma).
$$
\end{lemma}

\begin{proof}
Recall that $P$ and $D$ are both subsets of $\ring(c,r,2r)$. Thus we have $$\max\{\cost(P,c),\cost(D,c)\}\leq 2nr$$

As $\Gamma(\Cfar)\geq \epsilon n/k$, at least $\epsilon n/k$ points in $P$ must have connection cost at least $5kr/\epsilon^2-2r>4kr/\epsilon^2$. So we know that $\cost(P,C,\Gamma)\geq \epsilon n/k\cdot 4kr/\epsilon^2=4nr/\epsilon$. So by triangle inequality we know that $$
\big | \cost(P,C,\Gamma)-\cost(D,C,\Gamma)\big|\leq \cost(P,c)+\cost(D,c)\leq 4nr<\epsilon \cost(P,C,\Gamma).
$$
\end{proof}

\begin{lemma}[{\cite[Lemma 13]{cohen2019fixed}}]
\label{lemma:VincentCoreset}
Let $C\subseteq \mathbb{R}^d,|C|=k$ and $\Gamma : C \to \mathbb{R}_+$ be an assignment constraint such that $\sum_{u\in C}\Gamma(u)=n$.
Let $Q$ be a uniform sample of $P\subseteq \ring(c,r,2r)$
with size $m=\tilde{O}(\epsilon^{-3}\log \delta^{-1})$ and re-weighted by $\forall x\in Q, w_Q(x)=n/m$.
Then with probability $1-\delta$,
$$
    |\cost(Q,C,\Gamma)-\cost(P,C,\Gamma)|\leq \varepsilon nr.
$$
\end{lemma}

Lemma~\ref{lemma:VincentCoreset} is a concentration inequality for a fixed center set with capacity constraints, given by \cite{cohen2019fixed}. To show the coreset property holds for all possible center sets, we carefully construct a discretization $\mathcal{F}$ of centers and the assignment constraints.

\paragraph{Definition of $\mathcal{F}$.}
Let $N$ denote an $\epsilon r$-net of the ball $B(c,\frac{5kr}{\epsilon^2})$.
So $|N|\leq (\frac{k}{\epsilon})^{O(d)}$.
Let $t := \lceil\frac{5k^2}{\epsilon^3}\rceil$
and $H := \{i \cdot \frac{n}{t} \mid i = 0, 1, \ldots, t\}$ denote
the set of multiples of $\frac{n}{t}$ that do not exceed $n$. Let $N\times H$ denote the set of weighted points $x$ such that $x\in N$ and $w(x)\in H$.
We define $$
    \mathcal{F}:=
    \left\{(C,\Gamma)\mid C\subseteq N,|C|\leq k,\forall x\in C, \Gamma(x)\in H,\sum_{x\in C}\Gamma(x)=n\right\}.
$$
Note that $|\mathcal{F}|\leq \big(\frac{k}{\epsilon}\big)^{O(kd)}\cdot (\frac{k}{\epsilon})^{O(k)}$ and thus $\log |\mathcal {F}|=\tilde{O}(kd\log\frac{k}{\epsilon})$.

In the following lemma, we show that the coreset property on $\mathcal{F}$ implies coreset property
on every $k$-point center set $C\subseteq \mathbb{R}^d$ and assignment $\Lambda$ with $\Lambda(\Cfar)\leq \epsilon n/k$.
\begin{lemma}\label{lemma:ring_fair_dis}
Let $(C,\Lambda)$ be a $k$-point center set in $\mathbb{R}^d$ with assignment constraint $\Lambda$ such that $\Lambda(C)=n$ and $\Lambda(\Cfar)\leq \epsilon n/k$ where $\Cfar=\{u\in C \mid \dist(u,c)>5r/\epsilon^2\}$, then there exists a $k$-point center set $S$ with assignment constraint $(S,\Gamma)\in \mathcal{F}$ such that for every weighted set $Q\subseteq \ring(c,r,2r)$ with $w_Q(Q)=n$,
$$
\cost(Q,C,\Lambda)\in \big(1\pm O(\epsilon)\big)\big(\cost(Q,S,\Gamma)+\Delta(C)\big)\pm O(\epsilon n r)
$$
where $\Delta(C) := \sum_{u\in \Cfar} \dist(u,c)\cdot \Lambda(u) $.
\end{lemma}

\begin{proof}

Let $\Cclose := C\setminus \Cfar$. For every $u\in \Cclose$,
let $S(u)\in N$ be a net point such that $\dist\big(u,S(u)\big)\leq \epsilon r$.
Let $u^*\in \Cclose$ denote the center with largest capacity, namely, $u^*\in \mathrm{argmax}_{u\in \Cclose} \Lambda(u)$.
Clearly $\Lambda(u^*)\geq \frac{\Lambda(\Cclose)}{k}\geq \frac{n}{2k}$.

Recall that $t=\lceil\frac{5k^2}{\epsilon^3}\rceil$. For every $x\in \Cclose\setminus \{x^*\}$, we let $\Gamma\big(S(x)\big) := \lfloor \frac{t\Lambda(x)}{n}\rfloor\cdot \frac{n}{t}$.
We define
$$
    \Gamma\big(S(x^*)\big) := n-\sum_{x\in \Cclose\setminus \{x^*\}} \Gamma\big(S(x)\big).
$$

As all $\Gamma\big(S(x)\big)$'s are multiples of $\frac{n}{t}$ and sum up to $n$, we know that $(S,\Gamma)\in \mathcal{F}$. We are ready to prove the lemma.

To simplify the presentation, we observe that it suffices to assume $S=\Cclose$. To see this, recall that $\forall u\in \Cclose, \dist(u,S(u))\leq \epsilon r$, thus replacing every $u$ with $S(u)$ produces at most $\epsilon n r$ error, which is affordable.

To prove the upper bound, $\cost(Q,C,\Lambda)\leq (1+ O(\epsilon))\big(\cost(Q,S,\Gamma)+\Delta(C)\big)+\epsilon n r$, it suffices to construct an assignment $\sigma': Q\times C \to \mathbb{R}_{+}$ that is consistent with $\Lambda$ so as
$$
\cost^{\sigma'}(Q,C)\leq (1+ \epsilon)\big(\cost(Q,S,\Gamma)+\Delta(C)\big)+\epsilon n r.
$$

Recall that we have assumed w.l.o.g, $S=\Cclose$. By construction we know that $\Lambda(u)-\Gamma(u)\in [0,\frac{n}{t}]$ for $u\in \Cclose\setminus \{u^*\}$ and $\Gamma(u^*)-\Lambda(u^*)\in [0,2\epsilon n/k]$.

To construct $\sigma'$, we modify the optimal assignment corresponding to $\cost(Q,S,\Gamma)$. Specially, we arbitrarily disconnect $\Gamma(u^*)-\Lambda(u^*)$ mass of points from $Q$ to $u^*$ in $\Gamma(u^*)$ and distribute the mass to $\Cclose\setminus \{u^*\}$ and $\Cfar$ to satisfy the requirements $\Lambda$ on them. We claim that by doing this, the connection cost increases by at most $$
(1+\epsilon)\Delta(C)+\epsilon n r.
$$
To see this, we first observe that sending the matching mass from $\ring(c,r,2r)$ to $\Cfar$ always costs at most $(1+\epsilon)\Delta(C)$. On the other hand, as $\Cclose\subseteq B(c,\frac{5kr}{\epsilon^2})$, and we send at most $$
\sum_{u\in \Cclose\setminus \{u^*\}} \big(\Lambda(u)-\Gamma(u)\big)\leq \frac{kn}{t}
$$
additional mass to $\Cclose\setminus \{u^*\}$,
the cost in this part increases by at most $\frac{kn}{t}\cdot \frac{5kr}{\epsilon^2}\leq \epsilon n r$. Thus we have proved the upper bound.

\medspace

It remains to prove the lower bound, $\cost(Q,C,\Lambda)\geq \big(1- O(\epsilon)\big)\big(\cost(Q,S,\Gamma)+\Delta(C)\big)-O(\epsilon n r)$. Let $\sigma$ denote the optimal assignment for $\cost(Q,C,\Lambda)$, namely, $\cost(Q,C,\Lambda)=\cost^\sigma(Q,C)$.

Let $\Tfar=\sum_{q\in Q}\sum_{u\in \Cfar} \sigma(q,u)\dist(q,u)$ and $\Tclose=\sum_{q\in Q}\sum_{u\in \Cclose} \sigma(q,u)\dist(q,u)$. So $\cost(Q,C,\Lambda)=\Tfar+\Tclose$. We observe that $\Tfar\geq (1- \epsilon) \Delta(C)$. So we just need to prove $\Tclose\geq \big(1- O(\epsilon)\big) \cost(Q,S,\Gamma)-O(\epsilon n r)$. It suffices to construct an assignment $\pi:Q\times S\to \mathbb{R}_+$ that is consistent with $\Gamma$ and $$
\cost^{\pi}(Q,S)\leq \big(1+O(\epsilon)\big)\Tclose+O(\epsilon n r).
$$

To construct $\pi$, we modify $\sigma$. Specifically, we arbitrarily disconnect $\Lambda(u)-\Gamma(u)$ mass for every $u\in \Cclose\setminus \{u^*\}$ and disconnect all mass connecting to $\Cfar$, and send all those mass to $u^*$. We note that we have re-allocated at most $2\epsilon n/k$ mass.

Let $y^*\in Q$ be a point such that $\dist(y^*,u^*)\leq \frac{\sum_{q\in Q} \sigma(q,u^*)\dist(q,u^*)}{\Gamma(u^*)}$. Note that such $y^*$ exists as there is always some point that contributes at most the average. Thus by triangle inequality, for every $x\in Q$, \begin{eqnarray*}
\dist(x,u^*)&\leq& \dist(x,c)+\dist(c,y^*)+\dist(y^*,u^*)\\
&\leq& 2r+2r+\frac{\sum_{q\in Q} \sigma(q,u^*)\dist(q,u^*)}{\Gamma(u^*)}\\
&=&4r+\frac{\sum_{q\in Q} \sigma(q,u^*)\dist(q,u^*)}{\Gamma(u^*)}
\end{eqnarray*}

Thus we know that the re-allocation of mass increases the cost by at most
$$
\bigg(4r+\frac{\sum_{q\in Q} \sigma(q,u^*)\dist(q,u^*)}{\Gamma(u^*)}\bigg)\cdot \frac{2\epsilon n}{k}\leq O(\epsilon n r)+O(\epsilon \cdot \Tclose)
$$
where we have used the fact that $\Gamma(u^*)\geq \frac{n}{2k}$ and $\sum_{q\in Q} \sigma(q,u^*)\dist(q,u^*)\leq \Tclose$.

So we have constructed such $\pi$ and thus proved the lower bound.
\end{proof}

\begin{proof}[Proof of Theorem \ref{thm:ring_fair}] Replacing $\delta$ with $\frac{\delta}{|\mathcal{F}|}$ in Lemma \ref{lemma:VincentCoreset}.
By union bound, \Cref{lemma:VincentCoreset} and the fact that
the uniform sample has size $\tilde{O}(\frac{kd}{\epsilon^3}\cdot \log \delta^{-1})$,
we know that w.p. at least $1-\delta$, the coreset property holds for all $(C,\Gamma)\in \mathcal{F}$.
By \Cref{lemma:ring_fair_dis}, we further know that the coreset property holds for all $(C,\Gamma)$ such that $\Gamma(\Cfar)\leq \epsilon n/k$.
By \Cref{lemma:ring_fair_far}, we know that the coreset property also holds for those $(C,\Gamma)$ such that $\Gamma(\Cfar)> \epsilon n/k$.
\end{proof}

\subsection{$\epsilon$-Coresets for Capacitated and Fair \kMedian}

Combing Theorem~\ref{thm:ring_fair} and Theorem~\ref{thm:reduct_ring} in the way of Section~\ref{sec:additive_err}, we obtain the algorithm for constructing assignment-preserving $\epsilon$-coresets.

\begin{theorem}\label{thm:capC}
There is a near-linear time algorithm that takes a data set $P\subseteq \mathbb{R}^d$ and outputs an assignment-preserving $\epsilon$-coreset $D\subseteq P$ with size $|D|=\tilde{O}(\frac{k^3}{\epsilon^6})$ for \kMedian. In particular, this implies an $\epsilon$-coreset for Capacitated \kMedian.
\end{theorem}

\paragraph{Fair clustering} Suppose $P\subseteq \mathbb{R}^d$ is a fair \kMedian instance with groups $P_1, \ldots, P_l\subseteq P$.  Let $\Delta$ denote the number of combinations of groups that one data point can belong to.
We note that our assignment-preserving coresets (\Cref{def:assign_coreset})
matches the case $\Delta=1$.
Thanks to a reduction of~\cite{HJV19}, we can use our assignment-preserving $\epsilon$-coreset to construct $\epsilon$-coreset for fair \kMedian.

\begin{theorem}[{\cite[Theorem 4.3]{HJV19}}]
Suppose there is an algorithm that for any instance $P\subseteq \mathbb{R}^d$ with groups $P_1, \ldots, P_l$, constructs an assignment-preserving $\epsilon$-coreset for \kMedian with probability $1-\delta$ in time $T(|P|,\epsilon,\delta)$.
Then there is an algorithm $\mathcal{A}$ that for any fair \kMedian instance $P$
such that $P$ can be partitioned into $\Delta$ disjoint groups $P^{(1)},\ldots,P^{(\Delta)}$ where each $P^{(i)}$ consists of points that belong to the same combination of groups,
$\mathcal{A}$ constructs an $\epsilon$-coreset on $P$ for fair \kMedian with probability $1-\delta$ in time
$$
    \tilde{O}\left(\sum_{i=1}^{\Delta} T\big(|P^{(i)}|,O(\epsilon),O(\delta)\big)\right).
$$
\end{theorem}

We thus can prove the following theorem.

\begin{theorem}\label{thm:fairC}
There is a linear algorithm that constructs an $\epsilon$-coreset for fair \kMedian with size $\tilde{O}(\Delta\cdot \frac{k^3}{\epsilon^6})$.
\end{theorem}

\paragraph{Extension to \kMeans.}
We can obtain a
$\poly(k/\epsilon)$-sized assignment-preserving $(\epsilon,nr)$-coresets for \kMeans using a very similar argument as in the \kMedian case.
In particular, we replace Lemma~\ref{lemma:VincentCoreset} with a similar concentration inequality for capacitated \kMeans
that appears in (the full version of) \cite{cohen2019fixed},
stated as follows.

\begin{lemma}[Lemma 32 in the full version of \cite{cohen2019fixed}]
\label{lemma:VincentCoresetKmeans}
Let $C\subseteq \mathbb{R}^d,|C|=k$ and $\Gamma : C \to \mathbb{R}_+$ be an assignment constraint such that $\sum_{u\in C}\Gamma(u)=n$.
Let $Q$ be a uniform sample of $P\subseteq \ring(c,r,2r)$
with size $m=\tilde{O}(\epsilon^{-3}\log \delta^{-1})$ and re-weighted by $\forall x\in Q, w_Q(x)=n/m$.
Then with probability $1-\delta$,
$$
    |\cost_2(Q,C,\Gamma)-\cost_2(P,C,\Gamma)|\leq \varepsilon nr^2+\epsilon\cdot \cost_2(P,C,\Gamma).
$$
\end{lemma}
     \section{Coresets for \kzC via Uniform Shattering Dimension}
\label{sec:uniform_vc}

As discussed earlier, our new framework Theorem \ref{thm:reduct_ring} allows one to construct coresets merely via uniform sampling which is naturally captured by the uniform shattering dimension of the metric space. We present a few results via uniform shattering dimension.

\paragraph{Functional representation of distance functions.}
As in~\cite{DBLP:conf/stoc/FeldmanL11,FSS20,BJKW21}, we consider
functional representation of distance functions
$\calF = \{ f_x : X \to \mathbb{R}_+ \}_{x \in P}$ for a data set $P$.
Intuitively, each $f_x \in \calF$ corresponds to a data point $x \in P$,
and it intends to represent $\dist(x, \cdot)$.
However, the generality of
functional representation enables one to consider
alternative definitions of $f_x$, particularly $f_x = \dist'(x, \cdot)$
for some slightly-perturbed $\dist'$ from $\dist$.
This perturbation has been shown useful for obtaining small
coresets in several recent works~\cite{huang2018epsilon,BJKW21}.

In the following Definition~\ref{def:sdim}, we define shattering dimension
which is a key measure for the complexity of $\calF$.
In particular, it has a direct relation to the size of coresets (Theorem~\ref{thm:sdim2coreset}).

\begin{definition}[Shattering dimension]
    \label{def:sdim}
    Consider $\calF = \{ f_x : X \to \mathbb{R}_+ \}_{x \in P}$.
For $c \in X, r \geq 0$, define $B_\calF(c, r) := \{ f_x \in \calF : f_x(c) \leq r \}$.
    The shattering dimension of $\calF$, denoted as $\sdim(\calF)$,
    is defined as the smallest integer $t \geq 1$, such that
    \begin{align*}
        \forall \calH \subseteq \calF, |\calH| \geq 2,\qquad
            \left| \left\{
                B_\calF(c, r) \cap \calH : c \in Y, r\geq 0
            \right\} \right|
        \leq |\calH|^t.
    \end{align*}
\end{definition}

\begin{theorem}
    \label{thm:sdim2coreset}
    Let $c \in X, r > 0$ and let $R \subseteq \ring(c, r, 2r)$.
    If there exists a set of functions $\calF := \{ f_x : X \to \mathbb{R}_+\}_{x \in R}$
    such that
    \begin{align*}
        \forall x \in R, c \in X,\qquad
        \dist(x, c) \leq f_x(c) \leq (1 + \epsilon) \cdot \dist(x, c),
    \end{align*}
    then for every $0 < \epsilon, \delta < 1$,
    a uniform sample $S \subseteq R$ of size
    $\tilde{O}_z(\epsilon^{-2}) \cdot k \log \delta^{-1} \cdot \sdim(\calF)$
    with each point $p \in S$ reweighted by $w_S(p) := \frac{|R|}{|S|}$,
    is an $(\epsilon, \cost_z(R, c))$-coreset for \kzC on $R$ with probability at least $1 - \delta$.
\end{theorem}

Indeed, similar theorems that relate shattering dimension to coresets have been
discovered in the literature~\cite{DBLP:conf/stoc/FeldmanL11,FSS20,BJKW21}.
However, a fundamental difference is that ours work for \emph{uniform} shattering dimension,
or put it another way, shattering dimension of the range space of \emph{unweighted} balls,
while previous works require a \emph{universal} upper bound for $\sdim(\calF_v := \{ v(x) \cdot f_x \}_{x \in P})$
over \emph{all} $v : P \to \mathbb{R}_+$.
This requirement of uniform shattering dimension is much relaxed, and we shall see this immediately implies many completely-new and/or improved coreset results.

Technically, a key difference from previous arguments is that we do not,
and actually cannot, use the sensitivity sampling framework in the analysis.
In particular, it can be verified that even for a ring data set,
the sensitivity can still vary a lot between points,
which suggests that our proposed uniform sampling algorithm has to suffer a large error.
However, we reach this negative conclusion exactly because the sensitivity sampling
framework always aims for a strict/strong guarantee of multiplicative error,
while in our case we can actually accept an \emph{additive} error of $\epsilon \cdot \cost(R, c)$.
Note that this additive error can be very significant compared with the optimal solution of $R$,
since this $c$ is not necessarily a near-optimal center for $R$.
Hence, our analysis crucially charges the error from the uniform sampling to
this additive term.

Since we do not need the sensitivity framework and thus no sensitivity bound
is necessary,
it both simplifies the analysis and improves the coreset size.
In particular, there is usually an additional $k$ factor in the sensitivity bound
which must be multiplied in the coreset size, and we save this term completely.
The only $k$ comes from the shattering dimension analysis (Lemma~\ref{lemma:ksdim2sdim}).

\begin{proof}[Proof of Theorem~\ref{thm:sdim2coreset}]

    Since we need to work with centers of $k$ points,
    we need to define the $k$-extension of the function set:
    let $\calF^{(k)} := \{ f_x^{(k)} : X^k \to \mathbb{R}_+ \}_{x \in R}$,
    such that $\forall x \in R, C \in X^k$,
    $f_x^{(k)}(C) := \min_{c \in C}{f_x(c)}$.
    Clearly,
    \begin{equation}
        \label{eqn:distorsion}
        \forall C \in X^k, x \in R,\qquad \dist(x, C) \leq f_x^{(k)}(C) \leq (1 + \epsilon) \cdot \dist(x, C).
    \end{equation}
    We say $\mathcal{D} \subseteq \calF^{(k)}$ is
    an $\alpha$-approximation $(\alpha > 0)$ of $\calF^{(k)}$ if
    \begin{align*}
        \forall C \in X^k, r \geq 0, \qquad
        \left| \frac{|B_\mathcal{D}(C, r)|}{|\mathcal{D}|} - \frac{|B_{\calF^{(k)}}(C, r)|}{|\calF^{(k)}|} \right| \leq \alpha.
    \end{align*}
    In the following Lemma~\ref{lemma:approx2coreset}
    we relate $\alpha$-approximation to coresets with additive error.

    \begin{lemma}
        \label{lemma:approx2coreset}
        Suppose $\mathcal{D} \subseteq \calF^{(k)}$ is an $\frac{\epsilon}{O_z(1)}$-approximation of $\calF^{(k)}$.
        Let $D := \{ x : f_x^{(k)} \in \mathcal{D} \} \subseteq R$
        be the corresponding point set of the functional respresentation $\mathcal{D}$.
        Let every point of $D$ be weighted by $\frac{|R|}{|D|}$,
        then $D$ is an $\left(\epsilon, \cost_z(R, c)\right)$-coreset for \kzC on $R$.
    \end{lemma}
    \begin{proof}
       We note that it suffices to prove $D$ is an $\left(O_z(\epsilon),\cost_z(R,c)\right)$-coreset when $\mathcal{D}$ is an $\epsilon$-approximation of $\mathcal{F}^{(k)}$ as we can scale back $\epsilon$ in the end.

        Fix a center $C \in X^k$.
By the formula of integration by part we know that,
        \begin{align*}
            \sum_{f^{(k)}_x \in \calF^{(k)}} \left({f_x^{(k)}(C)}\right)^z
            &= \int_{t\geq 0} ( |\calF^{(k)}| - |B_{\calF^{(k)}}(C, t)| ) \cdot z\cdot t^{z-1} dt
        \end{align*}

        We observe that there is an interval $[L_1, L_2]$ such that
        $ B_{\calF^{(k)}}(C, t) = \emptyset$ when $t < L_1$,
        and $B_{\calF^{(k)}}(C, t) = \calF^{(k)}$ when $t > L_2$.
        Since $\mathcal{D} \subseteq \calF^{(k)}$, we have
        \begin{align*}
            |B_{\calF^{(k)}}(C, t)| = \frac{|\calF^{(k)}|}{|\mathcal{D}|} \cdot
            |B_{\mathcal{D}}(C, t)|
        \end{align*}
        when $t \notin [L_1, L_2]$.
        Moreover, because of the distortion bound of $f^{(k)}_x$'s,
        and that $\Diam(R) \leq 2r$, we have $L_2 - L_1 \leq (2 + O(\epsilon)) \cdot r$.
        Therefore, by the definition of $\epsilon$-approximation
        and the correspondence between $R$ and $\calF^{(k)}$
        as well as that between $D$ and $\mathcal{D}$, we have
        \begin{align*}
\sum_{f^{(k)}_x \in \calF^{(k)}}\left({f_x^{(k)}(C)}\right)^z
            &= \int_{t\geq 0} ( |\calF^{(k)}| - |B_{\calF^{(k)}}(C, t)| )\cdot z\cdot t^{z-1} dt \\
            &= \int_{t \in [L_1, L_2]} ( |\calF^{(k)}| - |B_{\calF^{(k)}}(C, t)| ) \cdot z\cdot t^{z-1}  dt \\
                &\quad + \int_{t \in \mathbb{R}_{\geq 0} \setminus [L_1, L_2]} ( |\calF^{(k)}| - |B_{\calF^{(k)}}(C, t)| ) \cdot z\cdot t^{z-1} dt  \\
            &\in \int_{t \in [L_1, L_2]} \left( |\calF^{(k)}| - \frac{|\calF^{(k)}|}{|\mathcal{D}|}
            \cdot |B_\mathcal{D}(C, t)| \pm \epsilon \cdot |\calF^{(k)}| \right) \cdot z\cdot t^{z-1} dt \\
                &\qquad + \int_{t \in \mathbb{R}_{\geq 0} \setminus [L_1, L_2]}
                \left(|\calF^{(k)}| - \frac{|\calF^{(k)}|}{|\mathcal{D}|} \cdot
            |B_{\mathcal{D}}(C, t)|\right) \cdot z\cdot t^{z-1} dt \\
            &\in \int_{t \geq 0} \left( |\calF^{(k)}| - \frac{|\calF^{(k)}|}{|\mathcal{D}|}
                \cdot |B_{\mathcal{D}}(C, t)|\right)dt
                \pm O(\epsilon) \cdot  |\calF^{(k)}| \int_{t\in [L_1,L_2]} z\cdot t^{z-1} dt \\
            &\in \left(\sum_{f_x^{(k)} \in \mathcal{D}} w_D(f_x^{(k)}) \cdot f_x^{(k)}(C) \right) \pm O(\epsilon )\cdot |R|\cdot (L_2^z-L_1^z)  \\
            &\in   \cost_z(D, C) \pm O(\epsilon) \cdot |R| \cdot 2^{z-1}\cdot (L_1^z+(L_2-L_1)^z) \\
            &\in
            \cost_z(D,C)\pm O_z(\epsilon)\cdot |R|\cdot (L_1^z+r^z)
            \\
            &\in \cost_z(D,C)\pm O_z(\epsilon)\cdot\left( \cost_z(R,C)+\cost_z(R,c)
            \right)
        \end{align*}
        where for the last inequality we have used the fact that $\cost_z(R,C)\geq |R|\cdot L_1^z$ and $\cost_z(R,c)\geq |R|\cdot r^z$.
        Combining this with the fact that
        $\sum_{f^{(k)}_x \in \calF^{(k)}}{f_x^{(k)}(C)} \in (1 \pm O_z(\epsilon)) \cdot \cost_z(R, C)$,
        we conclude that
        \begin{align*}
            |\cost_z(R, C) - \cost_z(D, C)| \leq O_z(\epsilon) \cdot (\cost_z(R, C) + \cost_z(R, c)).
        \end{align*}
        It remains to scale $\epsilon$ and
        this finishes the proof of Lemma~\ref{lemma:approx2coreset}.
    \end{proof}
    \begin{lemma}[\cite{vapnik1971uniform,DBLP:conf/soda/LiLS00}]
        \label{lemma:alpha_approx}
        For every $0 < \alpha < 1$,
        a uniform sample of size
        \begin{align*}
            O_z(\epsilon^{-2} \cdot (\sdim(\calF^{(k)}) \log \epsilon^{-1} + \log \delta^{-1}))
        \end{align*}
        from $\calF^{(k)}$
        is an $\alpha$-approximation for $\calF^{(k)}$ with probability at least $ 1 - \delta $.
    \end{lemma}
    \begin{lemma}[{\cite[Lemma 6.5]{DBLP:conf/stoc/FeldmanL11}, \cite[Claim 6.1]{huang2018epsilon}}]
        \label{lemma:ksdim2sdim}
        $\sdim(\calF^{(k)}) \leq k \cdot \sdim(\calF)$.
    \end{lemma}
    We finish the proof of Theorem~\ref{thm:sdim2coreset}
    by combining Lemma~\ref{lemma:approx2coreset},
    Lemma~\ref{lemma:alpha_approx} and Lemma~\ref{lemma:ksdim2sdim}.
    \end{proof}

\subsection{Improved Coresets for \onezC in Low-dimensional Euclidean Spaces}

We focus on the case of $k = 1$, and the goal is to show the existence of an $\epsilon$-coreset for \onezC
whose dependence of $\epsilon^{-1}$ is sub-quadratic in $\mathbb{R}^d$ for constant $d$.
To this end, we need to replace Lemma~\ref{lemma:approx2coreset} with the following
improved $\alpha$-approximation bound which is obtained via discrepancy theory.
Crucially, when $d$ is small the exponent of $\epsilon$ is strictly smaller than $2$.
\begin{theorem}[{\cite[Theorem 4.10]{chazelle2001discrepancy}}]
    \label{thm:discrepancy} Assume the VC-dimension of $\calF^{(k)}$ is $\vcdim$.
    For every $0 < \alpha < 1$,
    there exists an $\alpha$-approximation for $\calF^{(k)}$ (defined in the proof of Lemma~\ref{lemma:approx2coreset})
    of size
    $O\left(\alpha^{-(2 - 2 / (\vcdim + 1))} (\log \frac{1}{\alpha})^{2 - 1 / (\vcdim + 1)}\right)$
    where the big-O notation hides a polynomial factor of $\vcdim$.
\end{theorem}

Here, the VC-dimension is a related notion to $\sdim$, and they are only up to
a logarithmic factor to each other (i.e., $\vcdim \leq \sdim \log \sdim$).
We specifically use the VC-dimension version of Theorem~\ref{thm:discrepancy}
instead of using the log-factor conversion from $\sdim$,
since any constant factor matters in our application.
To proceed, let $\calF := \{ f_x(c) = \|x - c\|_2 \}_x$ simply represent the $\ell_2$ distance function,
then it is well known that $\vcdim(\calF) \leq d + 1$.
Therefore, following the proof of Theorem~\ref{thm:sdim2coreset},
plug in this $\vcdim(\calF)$ bound into Theorem~\ref{thm:discrepancy},
and combine with Lemma~\ref{lemma:approx2coreset},
we obtain the following theorem.
\begin{theorem} \label{thm:SR}
 For every integer $d \geq 1$, $c \in \mathbb{R}^d$, $r > 0$, ring data set
 $R \subseteq  \ring(c, r, 2r) \subseteq \mathbb{R}^d$,
 there is an $(\epsilon, \cost_z(R, c))$-coreset for \onezC on $R$,
 with size $\tilde{O}_z(\epsilon^{-(2 - 2 / (d + 1))})$ where the big-O notation hides a polynomial factor of $d$.
\end{theorem}

Finally, combining Theorem~\ref{thm:SR} with Theorem~\ref{thm:onez},
and the general reduction in Section~\ref{sec:additive_err},
we obtain the following corollary.
\begin{corollary}
    \label{cor:lowdim}
    For every integer $d \geq 1$, data set $P \subseteq \mathbb{R}^d$,
    there exists an $\epsilon$-coreset of size
    $\tilde{O}_z(\epsilon^{-(2 - 2 / (d + 1))})$ for \onezC on $P$ where the big-O notation hides a polynomial factor of $d$. In particular, if $P\subseteq \mathbb{R}^2$ there exists an $\tilde{O}(\epsilon^{-1.5})$-sized coreset for \oneMedian on $P$ and if $P\subseteq \mathbb{R}^3$ there exists an $\tilde{O}(\epsilon^{-1.6})$-sized coreset for \oneMedian on $P$.
\end{corollary}

\subsection{Coresets for Wasserstein Barycenter}
\paragraph{Wasserstein distance on a general metric space.}
Suppose $M(X, \dist)$ is an underlying metric,
and we define the $p$-Wasserstein distance with respect to $M$.
Let integer $\ell \geq 1$ denote the size of the support of a point in the Wasserstein metric/distance.
Let $\calX := X^\ell$ be the set of $\ell$-tuples of $X$.
The $p$-Wasserstein distance, denoted $\dWSp : \calX \times \calX \to \mathbb{R}_+$,
is defined on $\calX$, such that $\dWSp(S, T)$ is the cost of the min-cost matching between $S, T$,
where the costs are measured in $\dist$ to the power of $p$.
Formally, let $\Gamma$ be the set of all bijection between $\ell$ elements
(which may be interpreted as permutations of $[\ell]$),
and for $\pi \in \Gamma$, for $S, T \in \calX$,
let $\pi_{S,T}$ be the induced bijection from $S$ to $T$.
Then,
\begin{align*}
    \dWSp(S, T) := \min_{\pi \in \Gamma}{ \left(\sum_{x \in S}{\dist^p(x, \pi_{S, T}(x))}\right)^{1/p} },
\end{align*}

\paragraph{Wasserstein barycenter.}
The $p$-Wasserstein barycenter problem is \oneMedian on the metric $(\calX, \dWSp)$,
and the objective function is
\begin{align*}
    \forall c \in \calX,\qquad \cost^{(p)}_{\mathrm{WS}}(\calP, c) :=
    \sum_{S \in \calP}{\dWSp(S, c)}.
\end{align*}

\begin{theorem}
    \label{thm:sdim_p_ws}
    Consider $\calP \subseteq \calX$,
    and let $P := \bigcup_{S \in \calP}{S} \subseteq X$.
    If there exists $\calF = \{ f_x  : X \to \mathbb{R}_+ \}_{x \in P}$
    and $t \geq 1$ such that
    \begin{align*}
        \forall x \in P, y \in X, \qquad
        \dist(x, y) \leq f_x(y) \leq t \cdot \dist(x, y),
    \end{align*}
    then for every $p > 0$, there exits
    $\calF' = \{ f'_S : \calX \to \mathbb{R}_+ \}_{S \in \calP}$,
    such that
    \begin{align*}
        \forall S \in \calP, T \in \calX, \qquad
        \dWSp(S, T) \leq f'_S(T) \leq t \cdot \dWSp(S, T),
    \end{align*}
    and that
    \begin{align*}
        \sdim(\calF') \leq (\sdim(\calF) + 1) \cdot \ell.
    \end{align*}
\end{theorem}

\begin{proof}
    For $S \in \calP$, define
    \begin{align*}
        f'_S(T) := \min_{\pi \in \Gamma} {
            \left(\sum_{x \in S}{(f_x(\pi_{S, T}(x)))^p}\right)^{1 / p}
        }.
    \end{align*}
    By the distortion guarantee of $f_x(\cdot)$ with respect to $\dist(x, \cdot)$,
    we have
    \begin{align*}
        \forall S \in \calP, T \in \calX, \qquad
        \dWSp(S, T) \leq f'_S(T) \leq t \cdot \dWSp(S, T).
    \end{align*}

    Now we analyze the shattering dimension of $\calF'$.
    Fix $\calH \subseteq \calF'$ with $|\calH| \geq 2$.
    It suffices to give an upper bound for
    \begin{align*}
        |\{ B_{\calF'}(S, r) \cap \calH : S \in \calX, r \geq 0  \}|.
    \end{align*}
    If $|\calH| \leq \ell$, then
        $\left| \left\{
            B_{\calF'}(S, r) \cap \calH : S \in \calX, r\geq 0
        \right\} \right| \leq 2^{|\calH|} \leq 2^\ell \leq |\calH|^\ell$,
    which implies $\sdim(\calF') \leq \ell$.

    Otherwise, $|\calH| > \ell$, and we use the following argument.
    \begin{align*}
        &\quad\left| \left\{
            B_{\calF'}(S, r) \cap H : S \in \calX, r\geq 0
        \right\} \right| \\
        &=\left| \left\{
            \{ f'_T \in \calH : f'_T(S) \leq r \}: S\in \calX, r \geq 0
        \right\} \right| \\
        &=\left| \left\{
            \left\{ f'_T \in \calH : \min_{\pi \in \Gamma}
                { \sum_{x \in S}{(f_x(\pi_{S, T}(x)))^p} \leq r^p } \right\}
                : S \in \calX, r \geq 0
        \right\}\right| \\
        &\leq \sum_{\pi \in \Gamma}{
            \left|\left\{
                \left\{
                    f'_T \in \calH : \sum_{x \in S}{(f_x(\pi_{S, T}(x)))^p} \leq r^p
                \right\} : S \in \calX, r \geq 0
            \right\}\right|
        } \\
        &\leq \sum_{\pi \in \Gamma}{
            \left|\left\{
                \left\{
                    f'_T \in \calH : \bigwedge_{x_i \in S}{ (f_{x_i}(\pi_{S, T}(x_i)))^p \leq r_i^p }, \sum_{i} r_i^p = r^p
                \right\} : S \in \calX, r \geq 0
            \right\}\right|
        } \\
        &\leq \sum_{\pi \in \Gamma}{
            \left|\left\{
                \left\{
                    f'_T \in \calH : \bigwedge_{x_i}{ f_{x_i}(\pi_{S, T}(x_i)) \leq r_i }
                \right\} :  \forall i \in [\ell], x_i \in X \land r_i \geq 0
            \right\}\right|
        } \\
        &\leq |\Gamma| \cdot |\calH|^{\sdim(\calF) \cdot \ell}
        \leq \ell^\ell \cdot |\calH|^{\sdim(\calF) \cdot \ell}.
    \end{align*}
    Since we assume $|H| > \ell$, we have
    \begin{align*}
        \ell^\ell \cdot |\calH|^{\sdim(\calF) \cdot \ell}
        \leq |\calH|^{(\sdim(\calF) + 1) \cdot \ell},
    \end{align*}
    which concludes the proof.
\end{proof}

Combining Theorem~\ref{thm:sdim_p_ws} with Theorem~\ref{thm:onez} and Theorem~\ref{thm:sdim2coreset},
we conclude the following coreset bound for $p$-Wasserstein barycenter problem
on the Euclidean $\mathbb{R}^d$ metric space.
\begin{theorem}\label{thm:barycenter}
    There is an algorithm that given $\ell, p > 0$, $0 < \epsilon < 1$,
    integer $d \geq 1$, $\calP \subseteq (\mathbb{R}^d)^\ell$ and $ c \in (\mathbb{R}^d)^\ell$,
    computes a weighted subset $\mathcal{S} \subseteq \calP$
    such that the size of $S$ is bounded by $\tilde{O}(\epsilon^{-2}d \ell)$, which is independent of $p$,
    and $\mathcal{S}$ is an $(\epsilon, \cost^{(p)}_{\mathrm{WS}}(\calP, c))$-coreset for $p$-Wasserstein barycenter. Moreover, the algorithm 
    runs in time $\tilde{O}(|\mathcal{P}| \cdot T_{\dWSp})$, where $T_{\dWSp}$ is the time for
    evaluating the distance between a pair of points in $(\mathbb{R}^d)^\ell$.
\end{theorem}
\begin{remark}
    \label{remark:p_ws}
    We note that by combining with the iterative size reduction technique
    in a way similar as in the proof of Theorem~\ref{thm:ring_fair},
    one can improve the coreset bound in the above Theorem~\ref{thm:barycenter}
    to $\tilde{O}(\epsilon^{-4} \ell)$.
    Furthermore, since our shattering dimension bound works for a general metric space,
    we can also obtain small coresets for $p$-Wasserstein barycenter on an arbitrary underlying metric,
    as long as there exists a set of functions $\calF$ with
    low distortion and low shattering dimension for the underlying metric,
    as required in Theorem~\ref{thm:sdim_p_ws}.
\end{remark}

\subsection{More Coresets for \kMedian via Uniform Shattering Dimension}
\label{sec:more}

We present two more results that can be obtained
by plugging in known uniform shattering dimension bound into our framework.
Interestingly, whether the weighted shattering dimensions of both metric spaces are bounded is still open and in previous works~\cite{BJKW21,BR22} much effort has been put in bypassing it.

\paragraph{Clustering for minor-excluded graphs.} When the metric space $M=(X,d)$ is the shortest path metrics of an $H$-minor free graph, the weighted shattering dimension of the ball range space associated with $M$ is not known to be bounded and \cite{BJKW21} applied alternative approaches to obtain an $\epsilon$-coresets for \kMedian with size $\poly(k/\epsilon)\cdot f(|H|)$ for some at least doubly-exponential $f$.
However, the uniform shattering dimension bound of $M$ is known to be merely $O(|H|)$~\cite{DBLP:journals/dm/BousquetT15}
and thus combing with our \Cref{thm:reduct_ring} and \Cref{lemma:approx2coreset},
we immediately obtain an $\epsilon$-coreset of size $\poly(k/\epsilon)\cdot |H|$.
Comparing with the known result \cite{BJKW21},
our result is much simplified and has greatly improved the size dependency on the size of excluded minor.

\paragraph{Clustering for polygonal curves under Fr\'{e}chet distance.}
Consider the metric space of Euclidean polygonal curves
$M=(\mathbb{X}_m^d,\dist)$ where $\mathbb{X}_m^d:=(\mathbb{R}^d)^m$
is the set of polygonal curves in $\mathbb{R}^d$ with complexity at most $m$ and $\dist_F$ is (continuous) Fr\'{e}chet distance.
Let $P$ be a data set in $M$, the \klMedian problems aims to minimize
$$
\cost_F(P,S)=\sum_{p\in P} \min_{s\in S}\dist_F(p,s)
$$
over $S\subseteq \mathbb{X}_{\ell}^d,|S|=k$.
Note that $\ell$ is often smaller than $m$.

Define the range space
$\mathcal{R}_{\ell,d}=\{B(p,r)\cap \mathbb{X}_m^d \mid r>0,p\in \mathbb{X}_\ell^d\}$.
It is shown in
\cite[Theorem 8.5]{DBLP:journals/dcg/DriemelNPP21} that the shattering dimension of $\mathcal{R}_{\ell,d}$ is bounded by $O\big(d^2 \ell^2\log ( m)\big)$.
Thus combing with our \Cref{thm:reduct_ring} and \Cref{thm:sdim2coreset}, we obtain an $\epsilon$-coreset for \kMedian on $P$ with size $\tilde{O}\left(\epsilon^{-3}k^3d^2 \ell^2\log(m)\right)$.
Compared with recent results, our coreset has size independent of $n$ which improves over the $O(\log n)$ dependence in~\cite{BR22},
and has only logarithmic dependence in $m$ instead of $\poly(m)$ as in~
\cite{DBLP:journals/corr/abs-2104-12141}.

Finally, similar results can also be obtained for clustering under Hausdorff distance,
since a similar shattering dimension bound for Hausdorff distance is obtained in~\cite[Theorem 7.7]{DBLP:journals/dcg/DriemelNPP21} as well.
 
    \addcontentsline{toc}{section}{References}
    \bibliographystyle{alphaurl}
    \begin{small}	
        \bibliography{ref}
    \end{small}

    \begin{appendices}
      
\section{Implying Coresets without Additive Error}
\label{sec:additive_err}
By a standard argument (used in e.g.,~\cite{Chen09}), one can show that the guarantees of assignment-preserving $O(\epsilon, \cost_z(P, c))$-coreset
already suffices to imply assignment-preserving $\epsilon$-coresets \emph{without} additive error.
We sketch the argument here.

Assume that an oracle $\mathcal{A}(P, c, \epsilon)$ computes
an assignment-preserving $(\epsilon, \cost_z(P, c))$-coreset  for \kzC (see Definition~\ref{def:assign:add}) on $P$
for every $P \subseteq X$ and center $c \in X$,
and suppose the size of the coreset is $T(\epsilon, k, z)$.
Then this oracle can be used to efficiently compute an $\epsilon$-coreset
for weighted set $P$ without the additive error.

Specifically, one starts with finding an $(\alpha, \beta)$-approximation
(for some $\alpha, \beta \geq 1$) $C^* = \{c_i^*\}_i$ for \kzC on $P$.
Here,
for $\alpha, \beta \geq 1$ and a weighted set $P$,
a set of points $C^* \subseteq X$ is called an
$(\alpha, \beta)$-approximation for \kzC on $P$
if $|C^*| \leq \beta k$ and
\begin{align*}
    \cost_z(P, C^*) \leq \alpha \cdot \min_{C' \subseteq X, |C'| \leq k}{\cost_z(P, C')}.
\end{align*}
An $(2^{O(z)}, 1)$-approximation can be obtained for \kzC in time $\tilde{O}(nk)$ on a general metric space.
The algorithm for $z = 1, 2$ was given by~\cite{DBLP:journals/ml/MettuP04}
and it was noted in~\cite{huang2020coresets} that
the~\cite{DBLP:journals/ml/MettuP04} algorithm can be modified to work for a general $z$.

After the $(\alpha, \beta)$-approximation is obtained,
we partition $P$ into $\{P_i\}_i$ where $P_i$
is the set of points $p$ such that $c_i^*$ is nearest from $C^*$ to $p$,
and apply $\mathcal{A}(P_i, c_i^*, \epsilon)$ to construct an assignment-preserving
$(\epsilon, \cost_z(P_i, c_i^*))$-coreset $S_i$ for every $P_i$.
We define a new coreset $S := \bigcup_i S_i$ as the union of $S_i$'s.
Clearly, $|S| \leq O(\beta k) \cdot T(\epsilon, k, z)$,

To analyze the error of $S$,
for every $C \subseteq X$ such that $|C| \leq k$ with assignment constraint $\Gamma$, let $\sigma$ denote the optimal assignment from $P$ to $C$ that is consistent with $\Gamma$. Let $\Gamma_i$ denote the assignment constraint such that $\Gamma_i(c)=\sum_{p\in P_i} \sigma(p,c)$.
So we have that $$
\cost_z(P,C,\Gamma)=\sum_{i} \cost_z(P_i,C,\Gamma_i)=\cost^{\sigma}_z(P,C)=\sum_{i}\cost^{\sigma}_z(P_i,C).
$$
Recall that $S_i$ is an assignment-preserving $(\epsilon,\cost_z(P_i,c_i^*))$-coreset of $P_i$, so there exists an assignment $\pi_i:P_i\rightarrow C$ that is consistent with $\Gamma_i$,
$$
\cost_z^{\pi_i}(S_i,C)-\cost_z^{\sigma}(P_i,C)\leq \epsilon \big( \cost_z^{\sigma}(P_i,C)+ \cost_z(P_i,c_i^*)\big).
$$
Let $\pi:S\rightarrow C$ denote the assignment consistent with $\Gamma$ such that $\pi$ restricted on $P_i$ is $\pi_i$. So we have,
\begin{align*}
  \cost_z^{\pi}(S, C)-\cost_z^{\sigma}(P,C)
    &= \sum_i \big( \cost_z^{\pi_i}(S_i, C) - \cost_z^{\sigma}(P_i, C)\big)  \\
    &\leq  \sum_i\epsilon \big( \cost_z^{\sigma}(P_i,C)+ \cost_z(P_i,c_i^*)\big) \\
    &\leq \epsilon \cdot (\cost_z^{\sigma}(P, C) + \alpha \OPT(P)) \\
    &=\epsilon \cdot (\cost_z(P, C,\Gamma) + \alpha \OPT(P))\\
    &\leq \epsilon (\alpha+1)\cost_z(P,C,\Gamma)
\end{align*}
where $\OPT(P) = \min_{C' \subseteq X, |C'| \leq k}{\cost_z(P, C')}$
is the optimal objective value for \kzC on $P$. Note that $\cost_z(S,C,\Gamma)\leq \cost_z^{\pi}(S,C)$ and $\cost_z(P,C,\Gamma)=\cost_z^{\sigma}(P,C)$. So we know that,
$$
\cost_z(S,C,\Gamma)-\cost_z(P,C,\Gamma)\leq \epsilon (\alpha+1)\cost_z(P,C,\Gamma).
$$
Similarly, we have that
$$
\cost_z(P,C,\Gamma)-\cost_z(S,C,\Gamma)\leq \epsilon (\alpha+1)\cost_z(P,C,\Gamma).
$$
and conclude that
$$
|\cost_z(P,C,\Gamma)-\cost_z(S,C,\Gamma)|\leq \epsilon (\alpha + 1) \cdot \cost_z(P, C,\Gamma).
$$

Hence, we rescale $\epsilon$ by an $\alpha$ factor to make $S$ an $\epsilon$-coreset. In conclusion, after the rescaling,
$S$ is an $\epsilon$-coreset of size
$O(\beta k)\cdot T(\epsilon / (\alpha + 1), k, z)$,
and the number of oracle accesses to $\mathcal{A}$
(noting that $\mathcal{A}$ should use $\epsilon / (\alpha + 1)$ as the error parameter)
is $O(\beta k)$.     \end{appendices}
\end{document}